\documentclass[11pt]{article}       
\usepackage{geometry}               
\usepackage{amsmath,amssymb,bm}
\usepackage{amsthm}
\usepackage{mathtools}
\usepackage{xcolor}
\usepackage[linesnumbered, ruled]{algorithm2e}
\usepackage{cite}
\usepackage[hidelinks,bookmarks=false]{hyperref}
\usepackage{cleveref}
\usepackage{makecell}
\usepackage{booktabs}
\usepackage{tabularx}

\newtheorem{theorem}{Theorem}
\newtheorem{lemma}{Lemma}

\newtheorem{remark}{Remark}
\newtheorem{assumption}{Assumption}
\newtheorem{proposition}{Proposition}

\newcommand{\xpost}{\bar{x}}
\newcommand{\xprior}{\bar{x}^-}
\newcommand{\Xpostbf}{\Sigma_{x,t-1}}

\newcommand{\Xpost}{\Sigma_{x,t}}
\newcommand{\Xpostopt}{\Sigma_{x,t}^*}
\newcommand{\Xprior}{\Sigma_{x,t}^-}

\newcommand{\Xpostinit}{\Sigma_{x,0}}
\newcommand{\Xpriorinit}{\Sigma_{x,0}^-}

\newcommand{\Xpriorinitopt}{\Sigma_{x,0}^{-,*}}

\newcommand{\xnom}{\hat{x}_0^-}
\newcommand{\Xnom}{\hat{\Sigma}_{x,0}^-}
\newcommand{\Pdist}{\mathbb{P}}
\newcommand{\Qdist}{\mathbb{Q}}
\newcommand{\Qhat}{\hat{\mathbb{Q}}}
\newcommand{\ambset}{\mathbb{D}}
\newcommand{\Gauss}{\mathcal{N}}

\newcommand{\Bures}{\mathcal{B}}
\newcommand{\Wass}[2]{W_2(#1, #2)}

\newcommand{\real}[1]{\mathbb{R}^{#1}}
\newcommand{\symm}[1]{\mathbb{S}^{#1}}
\newcommand{\psd}[1]{\symm{#1}_{+}}
\newcommand{\pd}[1]{\symm{#1}_{++}}
\DeclareMathOperator{\Tr}{Tr}
\DeclareMathOperator{\Cov}{{Cov}}
\Crefname{assumption}{Assumption}{Assumptions}
\Crefname{algorithm}{Algorithm}{Algorithms}
\Crefname{figure}{Fig.}{Figs.}

\title{Residual-Aware Distributionally Robust EKF:\\ Absorbing Linearization Mismatch via Wasserstein Ambiguity%
\thanks{This work was supported in part by the Information and Communications Technology Planning and Evaluation
(IITP) grants funded by MSIT No. 2022-0-00124, No. 2022-0-00480 and No. RS-2021-II211343,
Artificial Intelligence Graduate School Program (Seoul National University),
the National Research Foundation of Korea (NRF) grant funded by MSIT No. RS-2026-25477173,
the Air Force Office of Scientific Research Grant (AFOSR) Grant AF FA9550-25-1-0274, the National Aeronautics and Space Administration (NASA) under Grant 80NSSC22M0070, 
the National Science Foundation (NSF) under Grants CMMI 2135925, CPS 2311085 and IIS 2331878, and the Higher Education and Science Committee of the RA under Grant 24FP-C017.}
} % 
\author{
Minhyuk Jang$^\dagger$, 
Jungjin Lee$^\dagger$, 
Astghik Hakobyan, 
Naira Hovakimyan, and
Insoon Yang%
\thanks{The first two authors contributed equally.}%
\thanks{M. Jang and N. Hovakimyan are with the Department of Mechanical Science and Engineering, Grainger College of Engineering, University of Illinois Urbana-Champaign, Urbana, IL, USA {\tt\small \{jang64, nhovakim\}@illinois.edu}. J. Lee and I. Yang are with the Department of Electrical and Computer Engineering and Automation and Systems Research Institute, Seoul National University, Seoul, South Korea {\tt\small \{jungbbal, insoonyang\}@snu.ac.kr}. A. Hakobyan is with the Center for Scientific Innovation and Education and the National Polytechnic University of Armenia, Yerevan, Armenia {\tt\small astghik.hakobyan@csie.am}.}
}

\date{}

\begin{document}

\maketitle

\begin{abstract}                          
The extended Kalman filter (EKF) is a cornerstone of nonlinear state estimation, yet its performance is fundamentally limited by noise-model mismatch and linearization errors. We develop a residual-aware distributionally robust EKF that addresses both challenges within a unified Wasserstein distributionally robust state estimation framework. The key idea is to treat linearization residuals as uncertainty and absorb them into an effective uncertainty model captured by a stage-wise ambiguity set, enabling noise-model mismatch and approximation errors to be handled within a single formulation. This approach yields a computable effective radius along with deterministic upper bounds on the prior and posterior mean-squared errors of the true nonlinear estimation error. The resulting filter admits a tractable semidefinite programming reformulation while preserving the recursive structure of the classical EKF. Simulations on coordinated-turn target tracking and uncertainty-aware robot navigation demonstrate improved estimation accuracy and safety compared to standard EKF baselines under model mismatch and nonlinear effects.
\end{abstract}

\section{Introduction}

State estimation is a central problem in control, robotics, and signal processing, where the goal is to infer the hidden state of a dynamical system from noisy and partial measurements. The accuracy of these estimates directly determines the quality of downstream feedback control, motion planning, and safe navigation. For linear systems, the Kalman filter achieves this optimally in the minimum mean-squared error (MMSE) sense~\cite{kalman1960}. For nonlinear systems, exact Bayesian filtering is generally intractable, and one resorts to approximations such as the extended Kalman filter (EKF), unscented Kalman filter, and particle filters~\cite{simon2006optimal}.

Among these, the EKF remains the method of choice due to its recursive structure and computational efficiency. Its performance is, however, fundamentally limited by two compounding failure modes: \emph{noise-model mismatch}, since true noise distributions are rarely known exactly and must be approximated from data or prior assumptions; and \emph{linearization error}, since each predict--update cycle discards higher-order Taylor remainder terms that accumulate over time and can lead to filter divergence.

These two failure modes have largely been studied in isolation. Robust filtering methods address noise-model and parameter uncertainty~\cite{simon2006optimal,shaked2001new,speyer2003optimization,whittle1981risk,boel2002robustness}, while nonlinear filtering approaches handle Taylor remainders within EKF-type updates~\cite{noack2010bounding,ding2024high}. Consequently, EKF-type estimators that jointly address both remain limited, particularly in output-feedback settings where estimation errors directly impact control and safety~\cite{lorenzetti2020simple,kogel2017robust}.

Distributionally robust state estimation (DRSE) has emerged as a principled framework for handling noise-model uncertainty. In the linear setting, existing works rely on either divergence-based~\cite{levy2012robust,zorzi2016robust,zorzi2017robustness} or Wasserstein ambiguity set formulations~\cite{NEURIPS_DRKF,nguyen2023bridging,han2024distributionally,kargin2024distributionally,jang2025steady,ding2025exactly,wang2021distributionally,si2023distributionally,feng2026sinkhorn}, both yielding tractable solutions with statistically meaningful guarantees. For nonlinear systems, however, DRSE remains underdeveloped; distributionally robust particle filters~\cite{wang2022distributionally} illustrate the potential, but their computational cost limits real-time applicability. Extending DRSE to nonlinear EKF-type estimators that jointly address noise-model uncertainty and linearization residuals thus remains an open challenge.

To address this, we embed state-dependent linearization residuals directly into a Wasserstein ambiguity set,  unifying robustness to noise-model mismatch and nonlinear approximation error within a single framework. This leads to a novel \emph{residual-aware distributionally robust EKF} (DR-EKF) that preserves the recursive EKF while admitting a tractable semidefinite programming (SDP) reformulation and formal performance guarantees for the true nonlinear system.

The main contributions of this paper are as follows:
\begin{itemize}
\item 
We reinterpret linearization residuals as uncertainty and absorb them into a Wasserstein ambiguity set, thereby unifying robustness to noise-model mismatch and nonlinear approximation error.

\item We construct a computable stage-wise ambiguity radius that guarantees feasibility of the true effective uncertainty law, and derive a deterministic recursive certificate for the prior and posterior mean-squared errors of the true nonlinear estimation error.

\item  We derive a tractable SDP reformulation of the DRSE problem and specialize it to the proposed ambiguity radius, yielding a DR-EKF algorithm that preserves the recursive EKF structure while accounting for residual-induced uncertainty.

\item  We demonstrate improved estimation accuracy and enhanced closed-loop safety, including substantial reductions in collision rates, under noise-model misspecification and nonlinear effects.
\end{itemize}

The remainder of the paper is organized as follows:~\Cref{sec:prelim} presents the preliminaries,~\Cref{sec:drekf} develops the proposed DR-EKF, ambiguity set construction, and recursive MSE certificate, and~\Cref{sec:exp} reports simulation results.

\section{Preliminaries}\label{sec:prelim}
\subsection{Notation}

For a random variable $z$, let $\mathcal{L}(z)$ and $\mathcal{L}(z\mid\mathcal{G})$ denote its law and conditional law given a $\sigma$-algebra $\mathcal{G}$, respectively. $\mathcal{P}_2(\mathbb{R}^d)$ denotes the set of probability measures on $\mathbb{R}^d$ with finite second moments. Let $\symm{n}$, $\psd{n}$, $\pd{n}$ denote the sets of symmetric, positive semidefinite, and positive definite matrices in $\mathbb{R}^{n\times n}$. For $A,B\in\symm{n}$, $A\succeq B$ means 
$A-B\in\psd{n}$, and $\lambda_{\min}(A)$ denotes the smallest eigenvalue of $A$. Finally, $\|\cdot\|$ denotes the Euclidean norm for vectors and the spectral norm for 
matrices.

\subsection{Problem Setup}

We consider the discrete-time nonlinear stochastic system
\begin{align}
x_{t+1} &= f(x_t,u_t) + w_t, \label{eq:dynamics} \\
y_t &= h(x_t) + v_t, \label{eq:measurement}
\end{align}
where $x_t \in \mathbb{R}^{n_x}$ is the state, $u_t \in \mathbb{R}^{n_u}$ is a known input, and $y_t \in \mathbb{R}^{n_y}$ is the measurement at time $t$. The maps $f: \mathbb{R}^{n_x} \times \mathbb{R}^{n_u} \to \mathbb{R}^{n_x}$ and $h: \mathbb{R}^{n_x} \to \mathbb{R}^{n_y}$ are assumed twice continuously differentiable.  The initial state $x_0 \in \mathbb{R}^{n_x}$, process noise $w_t \in \mathbb{R}^{n_x}$, and measurement noise $v_t \in \mathbb{R}^{n_y}$ are random vectors with unknown distributions $\Qdist_{x,0}, \Qdist_{w,t} \in \mathcal{P}_2(\mathbb{R}^{n_x})$ and $\Qdist_{v,t} \in \mathcal{P}_2(\mathbb{R}^{n_y})$, respectively. We assume throughout that $x_0$, $\{w_t\}_{t\ge 0}$, and $\{v_t\}_{t\ge 0}$ are mutually independent, and that each of $\{w_t\}_{t\ge 0}$ and $\{v_t\}_{t\ge 0}$ is temporally independent.

At each time $t \ge 0$, the estimator has access to the measurement history $Y_t := \{y_0,\dots,y_t\}$ with $\mathcal{Y}_t := \sigma(Y_t)$ denoting the associated information $\sigma$-algebra. For convenience, we define $\mathcal Y_{-1}:=\{\emptyset,\Omega\}$. The objective is to estimate the state $x_t$ based on the available information $\mathcal{Y}_t$. If the noise distributions were known exactly, a natural formulation would be the stage-wise conditional MMSE problem
\begin{equation}\label{eq:stagewise-mmse}
\min_{\psi_t \in \mathcal{F}_t}
\ \mathbb{E}\!\left[
\|x_t - \psi_t(Y_t)\|^2
\,\middle|\,
\mathcal{Y}_{t-1}
\right],
\end{equation}
where $\mathcal{F}_t :=
\left\{
\psi_t : (\mathbb{R}^{n_y})^{t+1} \to \mathbb{R}^{n_x}
\ \middle|\
\psi_t \text{ is measurable}
\right\}$
is the set of admissible estimators. The optimal estimator $\xpost_t = \psi_t(Y_t)$ is then the conditional mean of the state.

% For nonlinear systems, the MMSE estimator~\eqref{eq:stagewise-mmse} is generally intractable. In practice, it is therefore approximated using recursive methods such as the EKF, which linearizes the nonlinear dynamics and measurement models around the current estimate. However, this approximation introduces additional errors due to higher-order Taylor remainder terms arising from the local linearization. 

For nonlinear systems, the MMSE estimation problem~\eqref{eq:stagewise-mmse} is generally intractable. In practice, it is approximated using recursive methods such as the EKF, which linearizes the nonlinear dynamics and measurement models around the current estimate. However, this local approximation introduces additional errors due to higher-order Taylor remainder terms.

\subsection{Distributionally Robust State Estimation}

While EKF-type methods provide a tractable approximation to the MMSE estimator, their performance can degrade due to two fundamental sources of uncertainty: the true noise distributions $\Qdist_{x,0}$, $\Qdist_{w,t}$, and $\Qdist_{v,t}$ are typically unknown and must be approximated from finite data or modeling assumptions, and the local linearizations of $f$ and $h$ introduce higher-order Taylor remainder terms acting as state-dependent 
perturbations. Existing Wasserstein DRSE formulations address the first source for linear systems~\cite{NEURIPS_DRKF,nguyen2023bridging,jang2025steady}, but do not account for linearization-induced uncertainty. These considerations motivate a distributionally robust formulation of~\eqref{eq:stagewise-mmse} that addresses both.

To describe these effects in a unified way, we introduce the stacked noise vector
\begin{align*}
\epsilon_t :=
\begin{cases}
[x_0^{\top}, v_0^{\top}]^{\top}, & t = 0,\\
[w_{t-1}^{\top}, v_t^{\top}]^{\top}, & t \ge 1,
\end{cases}
\end{align*}
which collects all exogenous uncertainty entering the system at stage $t$. Under local linearization, the Taylor remainder terms act as additional uncertainty.
As we will show in the next section, these residuals can be absorbed into an effective stacked uncertainty $\tilde{\epsilon}_t=\epsilon_t+\Delta_t$, where $\Delta_t$ captures the linearization residuals. 

% The local error dynamics then depend on the centered effective uncertainty $\tilde{\epsilon}_t-\hat{\epsilon}_t$.

We therefore consider the DR-MMSE estimation problem
\begin{equation}\label{eq:stagewise-minimax}
\inf_{\psi_t \in \mathcal{F}_t}
\ \sup_{\Pdist_{\epsilon,t} \in \ambset_{\epsilon,t}}
\ \mathbb{E}_{\Pdist_{\epsilon,t}}
\!\left[
\|x_t - \psi_t(Y_t)\|^2
\,\middle|\,
\mathcal{Y}_{t-1}
\right],
\end{equation}
where $\Pdist_{\epsilon,t}$ is a candidate marginal distribution of $\epsilon_t$ and $\ambset_{\epsilon,t} \subset \mathcal{P}_2(\mathbb{R}^{n_x+n_y})$ is an ambiguity set encoding our uncertainty in the nominal stacked-noise law.

Strictly speaking, for $t\ge 1$ the law $\Pdist_{\epsilon,t}$ alone does not determine the conditional law of $x_t$; one must also specify the carried state law from stage $t-1$. In~\Cref{subsec:tractable_dr_ekf}, we make this precise through a local stage-wise surrogate that models $\mathcal L(e_{t-1},\epsilon_t\mid\mathcal Y_{t-1})$ directly. With this understanding,~\eqref{eq:stagewise-minimax} is the distributionally robust counterpart of~\eqref{eq:stagewise-mmse}.

\subsection{Wasserstein Ambiguity Sets}\label{subsec:ambset}

In DRSE, the ambiguity set $\ambset_{\epsilon,t}$ is used to capture uncertainty in the noise model. We adopt Wasserstein ambiguity sets centered at a nominal noise distribution, which quantify uncertainty through the distance between probability measures.
Specifically, the type-2 Wasserstein distance between two probability measures $\Pdist,\Qdist \in \mathcal{P}_2(\mathbb{R}^d)$ is defined as
\[
\Wass{\Pdist}{\Qdist}
:=
\left(
\inf_{\pi \in \Pi(\Pdist,\Qdist)}
\int_{\mathbb{R}^d \times \mathbb{R}^d}
\|z-z'\|^2 \, d\pi(z,z')
\right)^{1/2},
\]
where $\Pi(\Pdist,\Qdist)$ denotes the set of all couplings with marginals $\Pdist$ and $\Qdist$. Intuitively, $W_2$ measures the minimum transportation cost between probability distributions, making it a natural metric for quantifying deviations between the true and nominal noise laws~\cite{villani2009optimal}. In our setting, it is especially appealing because it is compatible with moment-based uncertainty and yields tractable DRO reformulations~\cite{kuhn2019wasserstein,mohajerin2018data}.

Using this metric, we define the stage-wise ambiguity set for the stacked noise $\epsilon_t \in \mathbb{R}^{n_x+n_y}$ as
\begin{equation}\label{eq:amb_set}
\ambset_{\epsilon,t}(\theta_t)
:=
\left\{
\Pdist \in \mathcal{P}_2(\mathbb{R}^{n_x+n_y})
\ \middle|\
\Wass{\Pdist}{\Qhat_{\epsilon,t}} \le \theta_t
\right\},
\end{equation}
where $\Qhat_{\epsilon,t} \in \mathcal{P}_2(\mathbb{R}^{n_x+n_y})$ is the prescribed nominal distribution of the stacked noise and $\theta_t \ge 0$ is the radius, which controls the trade-off between performance and robustness. 

We take the nominal distribution of $\epsilon_t$ to be Gaussian with independent blocks. In particular, for each $t\geq 0$, let $\Qhat_{\epsilon,t}=\Gauss(\hat\epsilon_t,\hat\Sigma_{\epsilon,t})$ with $\hat\Sigma_{\epsilon,t}\in\pd{n_x+n_y}$, where
\[
\Qhat_{\epsilon,0} = \Qhat_{x,0}^- \otimes \Qhat_{v,0}, \quad \Qhat_{\epsilon,t} = \Qhat_{w,t-1}\otimes \Qhat_{v,t}, \; t\geq 1,
\]
where $\Qhat_{x,0}^- = \Gauss(\xnom, \Xnom)$, $\Qhat_{w,t-1} = \Gauss(\hat{w}_{t-1},\hat{\Sigma}_{w,t-1})$, and $\Qhat_{v,t}=\Gauss(\hat{v}_t, \hat{\Sigma}_{v,t})$ denote the nominal laws of the initial state, process noise, and measurement noise, respectively.

\begin{remark}
Although the nominal stacked covariance is taken block diagonal, the ambiguity set is defined over the joint stacked law. Consequently, the resulting optimization problem implicitly allows for full covariance structures, including the cross block $\Sigma_{wv,t-1}$. Thus, the least-favorable process--measurement correlations can emerge without requiring a correlated nominal model.
\end{remark}

The ambiguity set~\eqref{eq:amb_set} forms the basis for the tractable residual-aware DR-EKF developed in the next section.

\section{Residual-Aware Approach for Distributionally Robust EKF}\label{sec:drekf}

In this section, we develop a residual-aware DR-EKF that incorporates linearization residuals into a Wasserstein ambiguity set, unifying robustness to noise-model mismatch and linearization error within the recursive EKF structure.
At a high level, we introduce the \emph{residual-as-uncertainty} principle, whereby nonlinear approximation error is treated as uncertainty and absorbed into a single ambiguity set.
The construction proceeds in three steps: $(i)$ derive an effective uncertainty model that captures both stochastic noise and Taylor residuals, $(ii)$ reformulate the DRSE problem as a tractable SDP with EKF structure, and $(iii)$ construct a computable ambiguity radius with a recursive MSE certificate.

\subsection{Exact Nonlinear Error Dynamics and Effective Uncertainty Model}\label{subsec:error_dynamics_effective_noise}

We first derive the nonlinear error dynamics and show how linearization residuals can be captured in an effective uncertainty term.
We adopt an EKF-style predict--update structure, linearizing $f$ and $h$ around the current estimate at each stage (e.g.,~\cite[Chapter 13]{simon2006optimal}). Given the posterior estimate $\xpost_{t-1}$ at 
time $t-1$, the filter forms the prior mean
\begin{equation}
\label{eq:dr-ekf-prior-mean}
\xprior_t = f(\xpost_{t-1}, u_{t-1}) + \hat{w}_{t-1},
\end{equation}
and computes the Jacobians
\[
A_{t-1} \coloneqq \frac{\partial f}{\partial x}(\xpost_{t-1},u_{t-1}), \qquad C_t \coloneqq \frac{\partial h}{\partial x}({\xprior_t}).
\]
Then, the prior and posterior estimation errors are given by
$e_t^- \coloneqq x_t - \xprior_t$ and $e_t \coloneqq x_t - \xpost_t$.
To isolate the effect of linearization, we introduce the exact Taylor remainder maps
\[
\begin{split}
r^f_{t-1}(e) &\coloneqq f(\xpost_{t-1}+e,u_{t-1}) 
- f(\xpost_{t-1},u_{t-1}) - A_{t-1}e,\\
r^h_t(e) &\coloneqq h(\xprior_t+e) - h(\xprior_t) - C_t e.
\end{split}
\]

Substituting $x_{t-1}=\xpost_{t-1}+e_{t-1}$ into~\eqref{eq:dynamics} and subtracting the nominal prediction~\eqref{eq:dr-ekf-prior-mean} produces the exact prior error dynamics
\begin{equation}
e_t^-=A_{t-1} e_{t-1} + (w_{t-1}-\hat w_{t-1}) + r^f_{t-1}(e_{t-1}).
\label{eq:error-dynamics-prior}
\end{equation}
Similarly, substituting $x_t=\xprior_t+e_t^-$ into~\eqref{eq:measurement} and subtracting the nominal predicted measurement $h(\xprior_t)+\hat v_t$ gives the innovation
\begin{equation}
\nu_t \coloneqq y_t - h(\xprior_t) - \hat v_t = C_t e_{t}^{-}+
(v_t-\hat v_t)+
r^h_t(e_t^-).
\label{eq:error-dynamics-meas}
\end{equation}

The exact nonlinear error dynamics, therefore, take the form of a locally linear model driven by effective uncertainty. Specifically, introducing
\[
\tilde{w}_{t-1}
\coloneqq
w_{t-1}+r^f_{t-1}(e_{t-1}),
\quad
\tilde{v}_t
\coloneqq
v_t+r_t^h(e_t^-),
\]
the dynamics~\eqref{eq:error-dynamics-prior} and~\eqref{eq:error-dynamics-meas} can be written as
\begin{equation*}
e_t^-=
A_{t-1}e_{t-1}
+\bigl(\tilde{w}_{t-1}-\hat{w}_{t-1}\bigr),\; \nu_t= C_t e_t^-
+\bigl(\tilde{v}_t-\hat{v}_t\bigr).
\end{equation*}
Introducing the effective stacked uncertainty variable,
\begin{equation}
\label{eq:effective_residual}
\tilde\epsilon_t \coloneqq
\begin{cases}
[x_0^{\top},\tilde{v}_0^{\top}]^{\top}, & t=0,\\
[\tilde{w}_{t-1}^{\top},\tilde{v}_t^{\top}]^{\top}, & t\ge 1,
\end{cases}
\end{equation}
the uncertainty in the local error equations is captured by the law $\mathcal{L}(\tilde\epsilon_t)$, which combines stochastic noise and Taylor residuals into a single uncertainty model. This representation forms the basis for treating linearization residuals as uncertainty. 
Moreover, $\tilde\epsilon_t$ differs from the nominal stacked noise $\epsilon_t$ only through the residual term $\Delta_t := \tilde\epsilon_t-\epsilon_t$. Thus, enlarging the Wasserstein ambiguity set by the magnitude of $\Delta_t$ suffices to contain the effective uncertainty law, as formalized in~\Cref{subsec:computable_radius}.

\subsection{Stage-wise SDP Reformulation Under Local Linear-Gaussian Surrogate}\label{subsec:tractable_dr_ekf}

Following standard EKF practice, we approximate the nonlinear filtering problem at each time stage $t$ by a local linear-Gaussian surrogate. Under this surrogate, the DRSE problem reduces to a tractable stage-wise distributionally robust linear MMSE problem. This subsection develops the baseline DRSE formulation; linearization residuals will be incorporated in the next step by treating them as uncertainty.

Conditioned on $\mathcal Y_{t-1}$, we approximate the nonlinear system by a local linear-Gaussian surrogate defined as follows:
\begin{enumerate}
\item[(i)] The conditional law of the state is approximated as Gaussian, $\mathcal{L}(x_t \mid \mathcal Y_{t-1}) \approx \Gauss(\xprior_t,\Xprior)$.
\item[(ii)] The nonlinear innovation is approximated by the affine model $\nu_t \approx C_t(x_t-\xprior_t) + (v_t-\hat v_t) := \hat{\nu}_t$.
\end{enumerate}
The exact conditional joint law factorizes as $\mathcal{L}(e_{t-1},\epsilon_t \mid \mathcal Y_{t-1}) = \mathcal{L}(e_{t-1}\mid \mathcal Y_{t-1}) \otimes \mathcal{L}(\epsilon_t)$, reflecting the independence of past estimation error and current noise. In the surrogate, we replace $\mathcal{L}(e_{t-1}\mid \mathcal Y_{t-1})$ by a Gaussian approximation and retain this product structure only at the surrogate level.

Under this surrogate, the stage-wise surrogate of DR-MMSE problem~\eqref{eq:stagewise-minimax} depends only on the first two moments of the stacked noise distribution and thus admits an exact finite-dimensional reformulation in terms of covariance matrices.

\begin{proposition}[Stage-wise SDP reformulation]
\label{prop:sdp}
Fix $t \ge 1$, and let $\Xpostbf$ be the posterior covariance carried from stage $t-1$. Under the local linear-Gaussian approximation, the stage-wise surrogate of DR-MMSE estimation problem~\eqref{eq:stagewise-minimax} is equivalent to the following finite-dimensional SDP:
\begin{align}
\begin{split}
\max_{\substack{\Xprior,\Xpost, \Sigma_{w,t-1},\\ \Sigma_{v,t},
\Sigma_{wv,t-1},Z_t}}
\quad & \Tr(\Xpost)\\
\mathrm{s.t.}\quad
& \Xprior = A_{t-1}\Xpostbf A_{t-1}^{\top} + \Sigma_{w,t-1}, \\
&
\begin{bmatrix}
\Xprior - \Xpost & T_t \\
T_t^{\top} & S_t
\end{bmatrix}
\succeq 0,\;
\begin{bmatrix}
\hat\Sigma_{\epsilon,t} & Z_t \\
Z_t^{\top} & \Sigma_{\epsilon,t}
\end{bmatrix}
\succeq 0,\\
& \Tr\bigl(\Sigma_{\epsilon,t} + \hat\Sigma_{\epsilon,t} - 2Z_t\bigr)\le \theta_t^2,\\
&\Sigma_{\epsilon,t} \succeq \lambda_{\min}(\hat{\Sigma}_{\epsilon,t})I_{n_x+n_y},\\
& \Sigma_{\epsilon,t}
=
\begin{bmatrix}
\Sigma_{w,t-1} & \Sigma_{wv,t-1} \\
\Sigma_{wv,t-1}^{\top} & \Sigma_{v,t}
\end{bmatrix}
\in \psd{n_x+n_y},\\
& \Xprior,\Xpost,\Sigma_{w,t-1}\in \psd{n_x},
\Sigma_{v,t}\in \psd{n_y}, \Sigma_{wv,t-1}\in \real{n_x\times n_y},
Z_t\in \mathbb{R}^{(n_x+n_y)\times (n_x+n_y)},
\end{split}
\label{eq:sdp_t}
\end{align}
where
\begin{align}
T_t &:= \Xprior C_t^{\top} + \Sigma_{wv,t-1}, \notag\\
S_t &:= C_t\Xprior C_t^{\top}
      + \Sigma_{v,t}
      + C_t\Sigma_{wv,t-1}
      + \Sigma_{wv,t-1}^{\top}C_t^{\top}. \notag
\end{align}
Moreover, the SDP~\eqref{eq:sdp_t} attains its optimum. If $\Xpostopt$ and $\Sigma_{\epsilon,t}^*$ are optimal for~\eqref{eq:sdp_t}, with associated blocks $(T_t^*,S_t^*)$, then the least-favorable stacked noise law for the surrogate is $\Gauss(\hat\epsilon_t,\Sigma_{\epsilon,t}^*)$, and the corresponding minimax-optimal posterior estimate $\xpost_t=\psi_t^*(Y_t)$ admits the innovation form
\begin{equation}\label{eq:dr_estimator}
\psi_t^*(Y_t)=\xprior_t+K_t^* \nu_t,
\end{equation}
where $K_t^* = T_t^*(S_t^*)^{-1}$ is the robust gain.
\end{proposition}

In words, the SDP~\eqref{eq:sdp_t} computes the least-favorable covariance matrices within the ambiguity set, producing a minimax-optimal Kalman-like update. While \cite[Lemma~1]{jang2025distributionally} considers the uncorrelated case,~\Cref{prop:sdp} extends this formulation to allow cross-covariance $\Sigma_{wv,t-1}$, induced by the joint Wasserstein ambiguity set, thereby capturing noise coupling.\footnote{The least-favorable distribution shares the nominal mean $\hat\epsilon_t$, since the quadratic loss depends only on second-order moments and mean deviations are not adversarially beneficial.} The proof adapts the primal--dual sandwich argument of \cite{nguyen2023bridging}; see Appendix~A.

\begin{remark}[Initial stage $t=0$]\label{rem:init_stage}
{
At $t=0$, the stacked noise is $\epsilon_0 = [x_0^\top,v_0^\top]^\top$, and there is no preceding propagation step. Accordingly, the prior covariance $\Xpriorinit$ is treated as a free decision variable rather than being specified by a propagation constraint. Under the local linear-Gaussian surrogate, the DR-MMSE estimation problem~\eqref{eq:stagewise-minimax} at $t=0$ is equivalent to the following finite-dimensional SDP:

\begin{align}
\begin{split}
\max_{\substack{\Xpriorinit,\Xpostinit,\\ \Sigma_{v,0},
\Sigma_{xv,0},Z_0}}
\quad & \Tr(\Xpostinit)\\
\mathrm{s.t.}\quad
&
\begin{bmatrix}
\Xpriorinit - \Xpostinit & T_0 \\
T_0^{\top} & S_0
\end{bmatrix}
\succeq 0,\;
\begin{bmatrix}
\hat\Sigma_{\epsilon,0} & Z_0 \\
Z_0^{\top} & \Sigma_{\epsilon,0}
\end{bmatrix}
\succeq 0,\\
& \Tr\bigl(\Sigma_{\epsilon,0} + \hat\Sigma_{\epsilon,0} - 2Z_0\bigr)\le \theta_0^2,\\
&\Sigma_{\epsilon,0} \succeq \lambda_{\min}(\hat{\Sigma}_{\epsilon,0})I_{n_x+n_y},\\
& \Sigma_{\epsilon,0}
=
\begin{bmatrix}
\Xpriorinit & \Sigma_{xv,0} \\
\Sigma_{xv,0}^{\top} & \Sigma_{v,0}
\end{bmatrix}
\in \psd{n_x+n_y},\\
& \Xpriorinit,\Xpostinit\in \psd{n_x},
\Sigma_{v,0}\in \psd{n_y}, \Sigma_{xv,0}\in \real{n_x\times n_y},
Z_0\in \mathbb{R}^{(n_x+n_y)\times (n_x+n_y)},
\end{split}
\label{eq:sdp_t0}
\end{align}

where
\begin{align*}
T_0 &:= \Xpriorinit C_0^{\top} + \Sigma_{xv,0}, \\
S_0 &:= C_0\Xpriorinit C_0^{\top}
      + \Sigma_{v,0}
      + C_0\Sigma_{xv,0}
      + \Sigma_{xv,0}^{\top}C_0^{\top}. 
\end{align*}
Moreover, the SDP~\eqref{eq:sdp_t0} attains its optimum. If $(\Xpriorinitopt,\Xpostinit^*,\Sigma_{v,0}^*,\Sigma_{xv,0}^*,Z_0^*)$ is optimal for~\eqref{eq:sdp_t0}, with associated blocks $(T_0^*,S_0^*)$, then the least-favorable stacked noise law for the surrogate is $\Gauss(\hat\epsilon_0,\Sigma_{\epsilon,0}^*)$, and the corresponding minimax-optimal estimator $\psi_0^*(Y_0)=\xpost_0$ admits the innovation form
\[
\psi_0^*(Y_0)=\xprior_0+K_0^* \nu_0,
\]
where $K_0^* = T_0^*(S_0^*)^{-1}$ is the robust gain.

}
\end{remark}

Applying~\Cref{prop:sdp} recursively results in a DR-EKF under the local linear-Gaussian surrogate. Building on this foundation, we then incorporate linearization residuals by designing the ambiguity set to capture both the noise model mismatch and nonlinear effects.
% The remaining step is to choose the ambiguity radius $\theta_t$. The next subsections construct a computable $\mathcal{Y}_{t-1}$-measurable radius that captures linearization-induced uncertainty and contains the effective uncertainty law $\mathcal{L}(\tilde\epsilon_t \mid \mathcal{Y}_{t-1})$.

\subsection{Bounding the Linearization Residuals}\label{subsec:residual_bounds}

The gap between the effective stacked uncertainty $\tilde\epsilon_t$ and the nominal noise $\epsilon_t$ is precisely the linearization residual $\Delta_t=\tilde\epsilon_t-\epsilon_t$. To quantify its effect on the ambiguity set, we first derive \emph{oracle} bounds on $\Delta_t$. These bounds are not directly computable, but they make explicit how linearization error enters the ambiguity radius and set up the computable construction in \Cref{subsec:computable_radius}.

\begin{assumption}[Local curvature bounds]\label{assump:local_curv}
There exist constants $L_f,L_h \ge 0$ such that, in a neighborhood of the filter trajectory, $\frac{\partial f}{\partial x}(\cdot,u)$ is Lipschitz in $x$ with constant $L_f$, uniformly over admissible inputs $u$, and $\frac{\partial h}{\partial x}(\cdot)$ is Lipschitz in $x$ with constant $L_h$.
\end{assumption}

\begin{lemma}[Quadratic remainder bound]\label{lem:quad_rem_bound}
Under~\Cref{assump:local_curv}, the Taylor remainder terms satisfy
\[
\|r^f_{t-1}(e)\|
\le
\frac{L_f}{2}\|e\|^2,
\qquad
\|r^h_t(e)\|
\le
\frac{L_h}{2}\|e\|^2,
\]
for all $e$ in the neighborhood specified in~\Cref{assump:local_curv}.
\end{lemma}
The result follows from standard Taylor remainder bounds under Lipschitz Jacobians. See Appendix~\ref{app:pf_lem1} for the proof.
\Cref{lem:quad_rem_bound} provides pointwise quadratic bounds on the Taylor remainder terms. To convert these into bounds suitable for Wasserstein radius construction, we impose a local moment-regularity condition ensuring that the fourth moments are controlled by second moments.

\begin{assumption}[Fourth-moment bound]\label{ass:fourth_moment}
There exist constants $\alpha_f,\alpha_h \ge 0$ such that, for all $t\ge 0$, 
\[
(\mathbb E[\|e_t^-\|^4])^{1/2}\le \alpha_h \mathbb E[\|e_t^-\|^2],
\] 
and, for all $t\ge 1$, \[(\mathbb E[\|e_{t-1}\|^4])^{1/2}\le \alpha_f \mathbb E[\|e_{t-1}\|^2].
\]
\end{assumption}

This assumption imposes a local uniform kurtosis bound along the filter trajectory. Combined with \Cref{lem:quad_rem_bound}, it converts the quadratic remainder bounds into second-moment bounds compatible with the Wasserstein ambiguity set.\footnote{More generally, bounded kurtosis suffices. Under standard EKF stability conditions~\cite[Th.~3.1]{reif1999stochastic}, such bounds hold locally with finite constants. In the experiments, we set $\alpha_f=\alpha_h=\sqrt{3}$.}

Combining \Cref{lem:quad_rem_bound} with \Cref{ass:fourth_moment} gives the second-moment bounds on the residual magnitudes.
\begin{lemma}[Oracle residual radii]\label{lem:oracle_residual_radii}
Under \Cref{assump:local_curv,ass:fourth_moment}, define
\[
\eta^f_{t-1}
\coloneqq
\frac{L_f}{2}\alpha_f \mathbb{E}[\|e_{t-1}\|^2],\]
for $t\ge 1$ and \[
\eta^{h}_t
\coloneqq
\frac{L_h}{2}\alpha_h \mathbb{E}[\|e_{t}^{-}\|^2],\] 
for $t\ge 0$.
Then
\begin{equation}\label{eq:oracle_rad}
\begin{aligned}
\bigl(\mathbb{E}[\|r^f_{t-1}(e_{t-1})\|^2]\bigr)^{1/2}
&\le
\eta^f_{t-1},
 && t\ge 1,
\\
\bigl(\mathbb{E}[\|r^h_t(e_t^{-})\|^2]\bigr)^{1/2}
&\le
\eta^{h}_t,
 && t\ge 0.
\end{aligned}
\end{equation}
\end{lemma}
Proof of~\Cref{lem:oracle_residual_radii} is given in Appendix~\ref{app:pf_lem2}.
These bounds quantify how residuals enter the effective uncertainty through second moments. However, both $\eta^f_{t-1}$ and $\eta^h_t$ depend on the unknown error moments and are therefore not directly computable, leading to a circular dependence.

\subsection{Computable Effective Radius and Recursive Distributionally Robust Certificate}
\label{subsec:computable_radius}

We construct a computable ambiguity radius that accounts for both nominal noise uncertainty and linearization effects, together with deterministic recursive bounds on the prior and posterior second moments of the true nonlinear estimation error. The key challenge is the circular dependence introduced above: \emph{the residual bounds depend on the estimation error, while the estimation error depends on the ambiguity radius.} 

To address this, we develop a recursive certificate construction based on deterministic error bounds. The approach proceeds in three steps: $(i)$ build a deterministic posterior MSE envelope for the linear-Gaussian surrogate, $(ii)$ show that any deterministic posterior error bound induces a computable effective radius together with a prior MSE bound, and $(iii)$ combine these with the mismatch recursion to certify the true posterior MSE.

We begin by formalizing the uncertainty in the nominal noise model and introducing deterministic bounds on the linearized dynamics and filter gains.

\begin{assumption}[Nominal ambiguity radius]
\label{ass:nominal_radius}
For each $t \ge 0$, the nominal radius $\theta_{\epsilon,t} \ge 0$ satisfies
\[
\mathcal{L}(\epsilon_t)
\in \ambset_{\epsilon,t}(\theta_{\epsilon,t}).
\]
\end{assumption}

\begin{assumption}[Deterministic envelopes]\label{ass:det_env}
There exist deterministic sequences $\{\bar a_t,\bar m_t,\bar k_t\}_{t\ge0}$ such that, for every $t\ge0$, $\|A_t\|\le \bar a_t$, $\|I-K_t^*C_t\|\le \bar m_t$, and $\|K_t^*\|\le \bar k_t$ almost surely.\footnote{On a bounded operating region, \Cref{assump:local_curv} yields deterministic bounds on $\|A_t\|$ and $\|C_t\|$. Moreover, the SDP constraints together with the spectral tube bounds~\cite[Proposition~1]{jang2025distributionally} imply upper and lower spectral bounds on $\Xprior$ and $S_t^*$, hence also a bound on $\|K_t^*\|=\|T_t^*(S_t^*)^{-1}\|$. One may then take $\bar m_t := 1+\bar k_t\bar c_t$ when $\|C_t\|\le \bar c_t$.}
\end{assumption}
These assumptions ensure that the true noise law lies within the ambiguity set and that the filter dynamics are uniformly bounded, enabling recursive error propagation.

\subsubsection{Oracle effective radius}

For each $t \ge 0$, we define the oracle effective radius as
\begin{equation}
\eta_{\epsilon,t}
\coloneqq
\sqrt{(\eta^f_{t-1})^2 + (\eta^{h}_t)^2},
\qquad
\theta_{\epsilon,t}^{\mathrm{eff}}
\coloneqq
\theta_{\epsilon,t} + \eta_{\epsilon,t},
\label{eq:effective-radius-oracle}
\end{equation}
where $\eta_{t-1}^f$ and $\eta_t^h$ are defined in~\eqref{eq:oracle_rad}, with $\eta_{-1}^f \coloneqq 0$.
This radius enlarges the nominal ambiguity set just enough to absorb linearization residuals, so that the effective noise $\tilde\epsilon_t$ lies in $\ambset_{\epsilon,t}(\theta_{\epsilon,t}^{\mathrm{eff}})$. However, $\eta_{\epsilon,t}$ depends on the unknown error moments and is therefore not directly computable.

\subsubsection{Surrogate posterior MSE envelope}

We first construct a deterministic posterior MSE envelope for the surrogate system, independent of the residual terms. To this end, we consider the linear-surrogate error dynamics:
\begin{align}
\hat e_t^- &:= A_{t-1}\hat e_{t-1} + (w_{t-1}-\hat w_{t-1}),\quad t\ge1\label{eqn:ehat_pred}\\
\hat e_t &:= (I-K_t^*C_t)\hat e_t^- - K_t^*(v_t-\hat v_t), \quad t\ge0\label{eqn:ehat},
\end{align}
with $\hat e_0^- := x_0 - \xnom$.
Because $\hat e_t$ evolves under the same gains as the true error but without nonlinear remainder terms, it serves as a tractable proxy for the true estimation error.

\begin{lemma}[Surrogate posterior MSE envelope]\label{lem:surrogate_posterior}
Suppose \Cref{ass:nominal_radius,ass:det_env} hold. Define $\bar s_0 := \bar m_0 \sqrt{\Tr(\Xnom)}+\bar k_0 \sqrt{\Tr(\hat{\Sigma}_{v,0})} + (\bar m_0 + \bar k_0)\theta_{\epsilon,0}$ and for $t \ge 1$,
\begin{equation}\label{eqn:s_t_bar}
\begin{split}
\bar s_t
\coloneqq
&\bar m_t\left(\bar a_{t-1}\bar s_{t-1}+\sqrt{\Tr(\hat{\Sigma}_{w,t-1})}\right) +
\bar k_t \sqrt{\Tr(\hat{\Sigma}_{v,t})} + (\bar m_t + \bar k_t)\theta_{\epsilon,t}.
\end{split}
\end{equation}
Then $\mathbb E[\|\hat e_t\|^2]\le \bar s_t^2$ for all $t\ge 0$.
\end{lemma}
Proof of~\Cref{lem:surrogate_posterior} is given in Appendix~\ref{app:pf_lem_surrogate}.
The quantity $\bar s_t$ provides a deterministic posterior MSE envelope for the surrogate system and will serve as the basis for constructing a computable bound on the true estimation error.

\subsubsection{Computable prior bound and effective radius}

We now use the deterministic posterior bound to construct a computable effective ambiguity radius. The key observation is that the oracle residual bounds can be upper-bounded using a deterministic posterior error bound.

\begin{lemma}[Computable one-step radius bound]
\label{lem:onepass_upper_bound}
Fix $t \ge 0$ and suppose~\Cref{assump:local_curv,ass:fourth_moment,ass:nominal_radius,ass:det_env} hold.
For $t \ge 1$, let $\bar V_{t-1} \in \mathbb{R}$ satisfy $\mathbb E[\|e_{t-1}\|^2] \le \bar V_{t-1}^2$.
Define $\bar\eta^f_{t-1} := \frac{L_f}{2}\alpha_f \bar V_{t-1}^2$
with $\bar\eta^f_{-1}:=0$, and let
\begin{equation}
\gamma_t
:=
\bar a_{t-1}\bar V_{t-1}
+\sqrt{\Tr(\hat{\Sigma}_{w,t-1})} + \theta_{\epsilon,t}
+\bar\eta^f_{t-1}, \; t \ge 1.
\label{eq:gamma}
\end{equation}
with $\gamma_0 := \sqrt{\Tr(\Xnom)} + \theta_{\epsilon,0}$.
Set $\bar\eta^h_t := \frac{L_h}{2}\alpha_h \gamma_t^2$.
Then $\eta^f_{t-1}\le \bar\eta^f_{t-1}$, $\mathbb E[\|e_t^-\|^2]\le \gamma_t^2$, $\eta^h_t\le \bar\eta^h_t$,
and therefore
\begin{equation}
\theta^{\mathrm{eff}}_{\epsilon,t}
\le
\bar\theta^{\mathrm{eff}}_{\epsilon,t}
:=
\theta_{\epsilon,t}
+
\sqrt{(\bar\eta^f_{t-1})^2+(\bar\eta^h_t)^2}.
\label{eq:effective-radius}
\end{equation}
\end{lemma}
Proof of~\Cref{lem:onepass_upper_bound} is given in Appendix~\ref{app:pf_lem4}.
\Cref{lem:onepass_upper_bound} shows that, given a deterministic posterior bound at stage $t-1$, one can compute both a valid ambiguity radius $\bar\theta_{\epsilon,t}^{\mathrm{eff}}$ and a bound on the prior error at stage $t$.

\subsubsection{Certificate recursion and main result}

We now combine the two lemmas to obtain a recursive certificate. The true error decomposes as $e_t = \hat{e}_t + \delta_t$, where $\delta_t := e_t - \hat e_t$ captures the mismatch due to Taylor residuals. Let $\delta_t^- := e_t^- - \hat e_t^-$ denote the corresponding prior mismatch.
Initialize $\rho_0^- := 0$ and, for each $t \ge 0$, set
\begin{equation}\label{eq:cert_recursion}
\rho_t=\bar m_t\rho_t^-+\bar k_t\bar\eta_t^h, \;\; \rho_{t+1}^-=\bar a_t\rho_t+\bar\eta_t^f, \;\; \bar V_t=\bar s_t+\rho_t,
\end{equation}
where $\bar\eta_t^f:=\frac{L_f}{2}\alpha_f\bar V_t^2$.
The next theorem establishes that $\bar V_t^2$ provides an upper bound on the true posterior MSE.

\begin{theorem}[Distributionally robust true MSE certificate]
\label{thm:dr_certificate}
Suppose \Cref{assump:local_curv,ass:fourth_moment,ass:nominal_radius,ass:det_env} hold, and let
$\{\bar\theta_{\epsilon,t}^{\mathrm{eff}}, \gamma_t, \bar s_t, \rho_t^-, \rho_t, \bar V_t\}_{t\ge 0}$
be constructed according to~\eqref{eqn:s_t_bar}--\eqref{eq:cert_recursion}. Then, for every $t \ge 0$:
\begin{enumerate}
\item[(i)] \emph{(Marginal feasibility)}
The true effective uncertainty law satisfies $\mathcal L(\tilde\epsilon_t) \in \ambset_{\epsilon,t}(\bar\theta_{\epsilon,t}^{\mathrm{eff}})$.

\item[(ii)] \emph{(Prior MSE)}
The true prior error satisfies ${\mathbb E[\|e_t^-\|^2] \le \gamma_t^2}$.

\item[(iii)] \emph{(Posterior MSE)}
The true posterior error satisfies $\mathbb E[\|e_t\|^2] \le \bar V_t^2$.
\end{enumerate}
\end{theorem}
Proof of~\Cref{thm:dr_certificate} is given in Appendix~\ref{app:pf_thm1}. 

\begin{remark}
The bounds in \Cref{thm:dr_certificate} hold locally along the filter trajectory, i.e., as long as the trajectory remains within the neighborhood where \Cref{assump:local_curv} holds.
\end{remark}

Together, \Cref{lem:surrogate_posterior,lem:onepass_upper_bound} provide the key ingredients: \Cref{lem:surrogate_posterior} establishes a deterministic posterior envelope, while \Cref{lem:onepass_upper_bound} yields a computable ambiguity radius and prior bound. \Cref{thm:dr_certificate} combines these with the mismatch dynamics to obtain a fully recursive certificate for the true nonlinear estimation error. This leads to the practical residual-aware DR-EKF in~\Cref{alg:DREKF_TAC}, summarized as a fully recursive and implementable filtering algorithm that preserves the structure of the classical EKF.

\begin{remark}
In practice, replacing the deterministic bounds with realized quantities such as $\|A_t\|$, $\|I - K_t^*C_t\|$, $\|K_t^*\|$, and $\sqrt{\Tr(\Xpostopt)}$ can yield a less conservative pathwise surrogate of the certificate.
\end{remark}

\begin{algorithm}[t]
\DontPrintSemicolon
\caption{Residual-aware DR-EKF}
\label{alg:DREKF_TAC}
\KwIn{$\hat x_0^-$, $\hat\Sigma_{x,0}^-$, $\{\hat w_t,\hat v_t,\hat\Sigma_{w,t},\hat\Sigma_{v,t},\theta_{\epsilon,t}\}_{t\ge0}$, $\{\bar a_t,\bar m_t,\bar k_t\}_{t\ge0}$, $L_f$, $L_h$, $\alpha_f$, $\alpha_h$}

Initialize $\bar x_0^- \leftarrow \hat x_0^-$, $\rho_0^- \leftarrow 0$, and $\bar\eta_{-1}^f \leftarrow 0$\;

% observe $y_0$, compute $C_0 \leftarrow \frac{\partial h}{\partial x}(\bar x_0^-)$, and run the $t=0$ update using the explicit initialization formulas and the initial-stage variant of~\eqref{eq:sdp_t} to obtain $\bar x_0$, $\Sigma_{x,0}$, $\bar V_0$, $\bar x_1^-$, and $\rho_1^-$\;

\For{$t=0,1,\dots$}{
    Observe $y_t$ and compute $C_t \leftarrow \frac{\partial h}{\partial x}(\bar x_t^-)$\;

    Compute $\gamma_t$ via~\eqref{eq:gamma} and set $\bar\eta_t^h \leftarrow \frac{L_h}{2}\alpha_h\gamma_t^2$\;

    Compute $\bar\theta_{\epsilon,t}^{\mathrm{eff}}$ via~\eqref{eq:effective-radius}\;

    Solve~\eqref{eq:sdp_t} with $\bar\theta_{\epsilon,t}^{\mathrm{eff}}$ (use~\eqref{eq:sdp_t0} for $t=0$) to obtain $\Xpostopt$ and $K_t^*$\;

    Update the posterior mean and covariance:
    $\bar x_t \leftarrow \bar x_t^- + K_t^*(y_t - h(\bar x_t^-) - \hat v_t)$,
    $\Sigma_{x,t} \leftarrow \Xpostopt$\;

    Update the mismatch bound:
    $\rho_t \leftarrow \bar m_t\rho_t^- + \bar k_t\bar\eta_t^h$\;

    Compute $\bar s_t$ via~\eqref{eqn:s_t_bar} and set $\bar V_t \leftarrow \bar s_t + \rho_t$\;

    Compute $A_t \leftarrow \frac{\partial f}{\partial x}(\bar x_t,u_t)$\;
    
    Set $\bar\eta_t^f \leftarrow \frac{L_f}{2}\alpha_f\bar V_t^2$\;

    Predict the prior mean:
    $\bar x_{t+1}^- \leftarrow f(\bar x_t,u_t)+\hat w_t$\;

    Propagate the prior mismatch bound:
    $\rho_{t+1}^- \leftarrow \bar a_t\rho_t + \bar\eta_t^f$\;
}
\end{algorithm}

\section{Experiment Results}\label{sec:exp}
We evaluate the proposed residual-aware DR-EKF on three benchmarks, which are designed to evaluate both estimation accuracy and the practical impact of the proposed uncertainty representation.\footnote{Code and additional results are available at \href{https://github.com/jangminhyuk/ResidualAware_DREKF}{project codebase}.}
The first is coordinated-turn (CT) radar tracking, which tests estimation performance under covariance misspecification and increasing nonlinearity. The second is obstacle-avoiding navigation with an uncertainty-aware model predictive controller (MPC), which evaluates whether the posterior covariance improves downstream safety decisions in a static environment.
The third is crowd-aware navigation on the ETH/UCY dataset with pedestrian trajectory prediction and uncertainty-aware MPC, which evaluates whether the proposed covariance representation improves safety in dynamic environments with moving pedestrians and static obstacles.
For clarity, the key experimental parameters for all three benchmarks are summarized in~\Cref{tab:exp_params}.\footnote{The filter runs fully online, with each predict--update cycle completing within the discretization step $\Delta t$, and thus is real-time deployable. While our experiments use Python with MOSEK solver, the SDP~\eqref{eq:sdp_t} admits efficient first-order solutions (e.g., Frank--Wolfe~\cite[Algorithms 1 and 2]{nguyen2023bridging}) for broader deployment.}

In all three scenarios, the DR-EKF uses the effective radius~\eqref{eq:effective-radius}.
The nominal radius $\theta_{\epsilon,t}$ can be calibrated from data using standard DRO techniques, such as bootstrap confidence sets over the Wasserstein ball~\cite{mohajerin2018data}. In the experiments, it was selected via a validation procedure minimizing time-averaged MSE across Monte Carlo runs. The curvature constants $L_f$ and $L_h$ satisfying~\Cref{assump:local_curv} were derived analytically, using the fact that the state remains within a bounded operating region.

\begin{table}[t!]
\centering
\caption{Summary of experimental parameters.}
\label{tab:exp_params}
\footnotesize
\setlength{\tabcolsep}{2pt}

\begin{tabularx}{\linewidth}{
p{1.6cm} 
>{\raggedright\arraybackslash}X 
>{\raggedright\arraybackslash}X 
>{\raggedright\arraybackslash}X}
\toprule
\textbf{Param.} & \textbf{CT Target Tracking} & \textbf{Safe Navigation} & \textbf{Crowd-Aware Navigation} \\
\midrule

$\mu_{x,0}$ 
& $[0,\,0,\,2,\,0,\,0.3]^\top$ 
& $[0,\,0,\,0]^\top$
& $[2.5,\,6.0,\,0.0]^\top$ \\

$\Xpriorinit$ 
& $\mathrm{diag}(0.04 I_2, 0.25 I_2, 0.0025)$
& $\mathrm{diag}(0.01 I_2, 0.001)$
& $\mathrm{diag}(0.01 I_2, 0.007615)$ \\

$\Sigma_{w,t}$ 
& $\mathrm{diag}(0.0001 I_2, 0.0025 I_2, 0.0004)$ 
& $\mathrm{diag}(0.008 I_2, 0.002)$
& $\mathrm{diag}(0.01 I_2, 0.007615)$ \\

$\Sigma_{v,t}$ 
& $\mathrm{diag}(0.0001, 0.25)$ 
& $\mathrm{diag}(0.04 I_3, 0.03)$
& $\mathrm{diag}(0.0144 I_3, 0.001218)$ \\

$\xnom$
& $\mu_{x,0}$ 
& $[0,\,0,\,0]^\top$
& $[2.5,\,6.0,\,0.0]^\top$ \\

$\Xnom$
& $0.1 \Xpriorinit$
& $\mathrm{diag}(0.01 I_2, 0.001)$
& $\mathrm{diag}(0.005 I_2, 0.002285)$ \\

$\hat{\Sigma}_{w,t}$ 
& $0.1 \Sigma_{w,t}$
& $\mathrm{diag}(0.002 I_2, 0.0005)$
& $\mathrm{diag}(0.005 I_2, 0.002285)$ \\

$\hat{\Sigma}_{v,t}$ 
& $0.1 \Sigma_{v,t}$ 
& $\mathrm{diag}(0.005 I_3, 0.0075)$
& $\mathrm{diag}(0.0072 I_3, 0.000366)$ \\

$\theta_{\epsilon,t}$ 
& $0.001$ 
& $0.25$
& $0.3$ \\

$(L_f,L_h)$ 
& $(0.3,0.2)$ 
& $(0.3,0.5)$
& $(0.5,0.5)$ \\

$(\alpha_f,\alpha_h)$ 
& $(\sqrt{3},\sqrt{3})$ 
& $(\sqrt{3},\sqrt{3})$
& $(\sqrt{3},\sqrt{3})$ \\

\bottomrule
\end{tabularx}
\end{table}

\subsection{Coordinated-turn Target Tracking}\label{sec:target_tracking}

We consider a standard CT target tracking model~\cite{gustafsson1996best,roth2014ekf} with sampling time $\Delta t = 0.2 \, \mathrm{s}$:
\begin{equation}\label{eq:ct_dynamics}
f(x_t)=
\begin{bmatrix}
p_{x,t} + \frac{\sin(\omega_t \Delta t)}{\omega_t} v_{x,t}
- \frac{1-\cos(\omega_t \Delta t)}{\omega_t} v_{y,t} \\
p_{y,t} + \frac{1-\cos(\omega_t \Delta t)}{\omega_t} v_{x,t}
+ \frac{\sin(\omega_t \Delta t)}{\omega_t} v_{y,t} \\
\cos(\omega_t \Delta t) v_{x,t}
- \sin(\omega_t \Delta t) v_{y,t} \\
\sin(\omega_t \Delta t) v_{x,t}
+ \cos(\omega_t \Delta t) v_{y,t} \\
\omega_t
\end{bmatrix},
\end{equation}
where $x_t = [p_t^\top, v_t^\top, \omega_t]^\top$ comprises the planar position $p_t = [p_{x,t}, p_{y,t}]^\top$, planar velocity $v_t = [v_{x,t}, v_{y,t}]^\top$, and turn rate $\omega_t$. The target is observed via range and bearing:
\[
h(x_t)=
\begin{bmatrix}\|p_t\|, & \operatorname{atan2}(p_{y,t},p_{x,t})\end{bmatrix}^\top.
\]
We assume the true noise is zero-mean Gaussian, with nominal and true models differing only in their covariances (see~\Cref{tab:exp_params}). Performance is evaluated over $100$ Monte Carlo runs with horizon $T=50$.

\begin{figure}[t!]
    \centering
    \includegraphics[width=\linewidth]{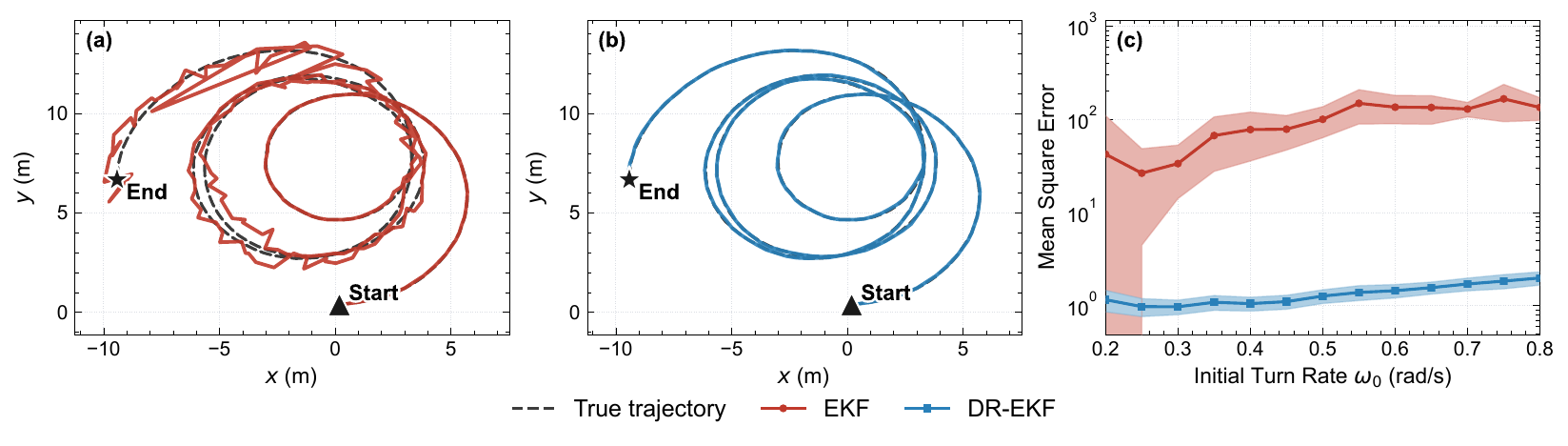}
    \caption{Coordinated-turn experiment under nominal covariance misspecification. (a) EKF with nominal statistics and (b) DR-EKF trajectories for the Monte Carlo run whose EKF MSE is the median over 100 runs. (c) Average MSE versus the initial turn rate $\omega_0$, where increasing $|\omega_0|$ corresponds to stronger nonlinearity; shaded regions denote $0.1$ times the standard deviation over 100 runs.}
    \label{fig:ct_combined}
\end{figure}

\Cref{fig:ct_combined} (a) and (b) show the EKF and DR-EKF state-estimate trajectories, respectively, for the same Monte Carlo run, selected as the one whose EKF MSE is the median among 100 runs. The results show that the EKF trajectory deviates more noticeably from the true trajectory, whereas the DR-EKF tracks it more closely and smoothly.

To examine the effect of nonlinearity, we vary the initial turn rate $\omega_0$ while keeping all covariance matrices and the nominal model fixed. As $|\omega_0|$ increases, the trigonometric couplings in~\eqref{eq:ct_dynamics} amplify linearization error. As shown in \Cref{fig:ct_combined} (c), EKF performance degrades as nonlinearity increases, whereas the DR-EKF maintains substantially lower MSE even at large $|\omega_0|$.
This indicates that the proposed DR-EKF is more robust to accumulated linearization error.

\subsection{Safe Navigation with Uncertainty-Aware MPC}\label{subsec:safe_nav_static}

We evaluate whether the posterior uncertainty of the residual-aware DR-EKF improves closed-loop safety in a goal-reaching task with an obstacle. The controller uses the estimated state together with a covariance-derived adaptive safety margin, linking uncertainty to safety constraints.

The system follows unicycle dynamics with sampling time $\Delta t = 0.2\,\mathrm{s}$, while the measurement model consists of three ultra-wideband range measurements and a heading measurement.
Specifically,
\begin{equation}\label{eqn:uni_model}
f(x_t,u_t)=
\begin{bmatrix}
p_{x,t} + s_t \cos(\psi_t)\Delta t\\
p_{y,t} + s_t \sin(\psi_t)\Delta t\\
\psi_t + \omega_t \Delta t
\end{bmatrix}, \qquad
h(x_t)=
\begin{bmatrix}
\|p_t - b_1\|_2 \\ \|p_t - b_2\|_2 \\ \|p_t - b_3\|_2 \\ \psi_t
\end{bmatrix},
\end{equation}
where $x_t = [p_t^\top,\psi_t]^\top$ is the robot state, with $p_t = [p_{x,t}, p_{y,t}]^\top$ denoting the planar position and $\psi_t$ the heading angle. The control input is $u_t = [s_t,\omega_t]^\top$, comprising the commanded linear and angular velocities, and $b_i \in \mathbb{R}^2$, $i=1,2,3$, are fixed beacon locations. 
We assume the true noise follows a zero-mean Gaussian distribution (see Table~\ref{tab:exp_params}).

\begin{figure}[!t]
    \centering
    \includegraphics[width=0.95\linewidth]{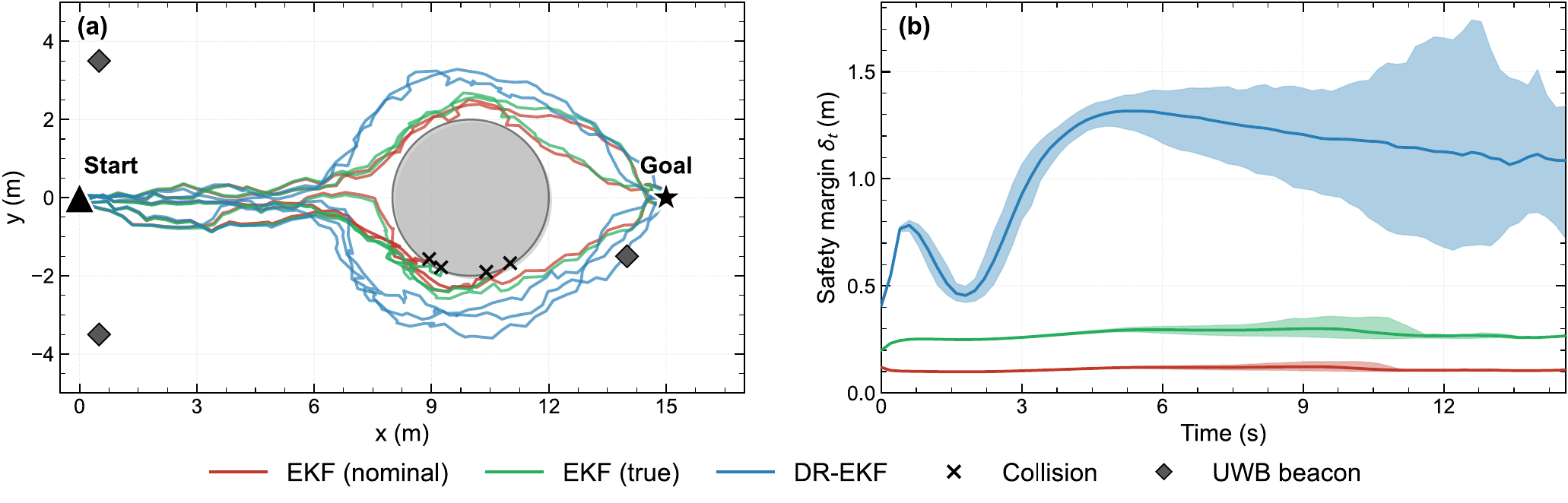}
    \caption{Closed-loop navigation under uncertainty-aware MPC. (a) Representative trajectories for EKF with nominal statistics, EKF with true noise statistics, and the proposed DR-EKF; crosses indicate collisions, and five of the 100 Monte Carlo runs are shown for clarity. (b) Time evolution of the safety margin $\delta_t$; shaded bands denote $\pm1$ empirical standard deviation over 100 runs.}
    \label{fig:safe_nav_combined}
\end{figure}

Starting from an initial state $x_0$, the robot must reach a predefined goal position $p_{\mathrm{goal}}$. At each time step, we solve a receding-horizon MPC problem based on the current state estimate $\xpost_t$. Let $\bar{x}_{t+k|t} = [\bar{p}_{t+k|t}^\top,\,\bar{\psi}_{t+k|t}]^\top$ denote the predicted state at stage $t+k$. The MPC problem is given by
\begin{subequations}
\begin{align}
\min_{\{\bar{x}_{t+k|t}\}_{k=0}^{N},\, \{u_{t+k|t}\}_{k=0}^{N-1}}
\quad &
\sum_{k=0}^{N-1}
\!\left(
q \|\bar{p}_{t+k|t}-p_{\mathrm{goal}}\|_2^2
+ r_s s_{t+k|t}^2
+ r_\omega \omega_{t+k|t}^2
\right)
+ q_f \|\bar{p}_{t+N|t}-p_{\mathrm{goal}}\|_2^2 \\
\mathrm{s.t.}\quad
& \bar{x}_{t|t} = \xpost_t,  \\
& \bar{x}_{t+k+1|t} = f(\bar{x}_{t+k|t}, u_{t+k|t}),  \qquad k=0,\dots, N-1,\\
& 0 \le s_{t+k|t} \le s_{\max}, \qquad
|\omega_{t+k|t}| \le \omega_{\max}, \quad k=0,\dots,N-1,  \\
& \|\bar{p}_{t+k|t} - p_{\mathrm{obs}}\|_2 \ge r_{\mathrm{obs}} + \delta_t,\label{eq:safety_constraint}
\quad k=1,\dots,N,
\end{align}
\label{eq:mpc_static}
\end{subequations}
where $q=5$, $r_s=0.5$, $r_\omega=0.5$, and $q_f=20$ are the cost weights, and the input constraints are $s_{\max}=1.5\,\mathrm{m/s}$ and $\omega_{\max}=2.0\,\mathrm{rad/s}$.

At each time step, the MPC enforces the safety constraint~\eqref{eq:safety_constraint}, where the safety margin is defined as 
\[
\delta_t := \kappa_\sigma \sqrt{\Tr(\Sigma_{p,t})},
\] 
with $\Sigma_{p,t}$ denoting the position block of $\Xpost$ and $\kappa_\sigma = 1.645$. Thus, the obstacle is adaptively inflated according to the current state-estimation uncertainty. Unlike fixed heuristic margins, this construction directly leverages the posterior covariance from the residual-aware DR-EKF, enabling the safety margin to adjust online to both distributional mismatch and linearization residuals.

We compare the EKF with nominal statistics, the EKF with true noise statistics, and the proposed DR-EKF. Over $100$ Monte Carlo runs, the collision rates are $77\%$, $43\%$, and $1\%$, respectively. As shown in \Cref{fig:safe_nav_combined}, the DR-EKF maintains greater obstacle clearance through a more adaptive safety margin, particularly near the obstacle, substantially reducing collisions and demonstrating the practical impact of the proposed uncertainty representation.

%%%%%%%%%%%%%%%%%%%%%%%%%%%%%%%%%%%%%%%%%%%%%%%%%%%%%%%%%%%%%%%%%%%%%%%%%%%%%%%%

\subsection{Safe Navigation in Crowd Scenes}

\begin{figure}[!t]
    \centering
    \includegraphics[width=0.95\linewidth]{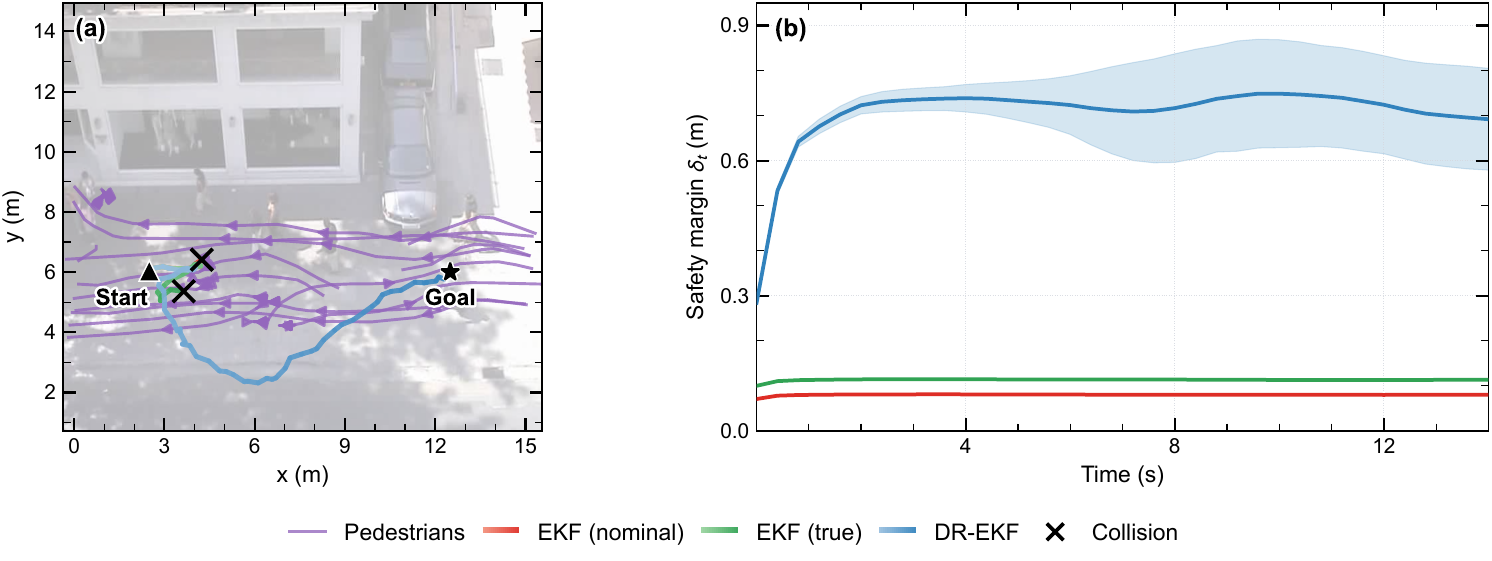}
    \caption{Closed-loop navigation under uncertainty-aware MPC with pedestrian prediction on the ZARA02 scene from the ETH/UCY dataset. (a) Representative trajectories for EKF with nominal statistics, EKF with true noise statistics, and the proposed DR-EKF; crosses mark collisions. (b) Time evolution of the safety margin $\delta_t$; shaded bands indicate $\pm1$ empirical standard deviation across 50 runs.}
    \label{fig:safe_nav_ped}
\end{figure}

We evaluate the DR-EKF in a crowd navigation setting where a mobile robot must reach a goal while avoiding pedestrians with uncertain predicted trajectories, in addition to static obstacles (e.g., walls, pillars). This setting provides a realistic and challenging testbed for uncertainty-aware navigation. Experiments are conducted on the ZARA02 sequence from the ETH/UCY dataset~\cite{lerner2007crowds,pellegrini2009you}, which contains 1052 frames.

The robot follows the same unicycle dynamics~\eqref{eqn:uni_model} with sampling time $\Delta t = 0.4\,\mathrm{s}$, using the same range-heading observation model with beacons at $b_1=[1.0,1.0]^\top$, $b_2=[14.0,1.0]^\top$, and $b_3=[7.5,9.0]^\top$. Noise statistics are given in~\Cref{tab:exp_params}. The environment includes static obstacles with known center $p_{\mathrm{obs}} \in \mathbb{R}^2$ and radius $r_{\mathrm{obs}}$. The robot, with body radius $r_{\mathrm{robot}} = 0.4\,\mathrm{m}$, starts from an initial state $x_0$ and must reach a predefined goal position $p_{\mathrm{goal}}$ while interacting with pedestrians replayed from the dataset.

Let $\mathcal{A}_t = \{1,\dots,M_t\}$ denote the set of pedestrians observed at time $t$. For each $j \in \mathcal{A}_t$, let $p^{(j)}_{t-H+1:t} = (p^{(j)}_{t-H+1}, \dots, p^{(j)}_{t})$, with $p^{(j)}_{\tau} \in \mathbb{R}^2$, denote the most recent $H = 8$ observed positions. A pedestrian prediction module $\mathcal{G}_{\mathrm{pred}}$, implemented by Trajectron++~\cite{salzmann2020trajectron++}, produces an $N = 12$ step forecast
\begin{equation}\label{eq:pedestrian_predictor}
\big(\tilde{p}^{(j)}_{t+k|t}\big)_{\substack{j\in\mathcal{A}_t\\k=1,\dots,N}}
=
\mathcal{G}_{\mathrm{pred}}
\!\left(
\big(p^{(j)}_{t-H+1:t}\big)_{j\in\mathcal{A}_t}
\right)
\end{equation}
where $\tilde{p}^{(j)}_{t+k|t} \in \mathbb{R}^2$ denotes the predicted mean position of pedestrian $j$ at time $t+k$.

The adaptive safety margin $\delta_t$ is defined in~\Cref{subsec:safe_nav_static}, with $\kappa_{\sigma} = 1$. The effective minimum clearance is then given by
\begin{equation}
d_{\min,t}^{\mathrm{eff}} = d_{\min}^{\mathrm{base}} + \delta_t,
    \label{eq:safety_margin}
\end{equation}
where $d_{\min}^{\mathrm{base}} = r_{\mathrm{robot}} + r_{\mathrm{ped}}$ is the nominal physical clearance. Here,
$r_{\mathrm{ped}} = 0.1/\sqrt{2}\,\mathrm{m}$ is a conservative half-width approximation of a pedestrian. Thus, the safety margin adapts online to the state-estimation uncertainty while remaining grounded in physical geometry.

The MPC controller follows the formulation in~\eqref{eq:mpc_static}, with predicted pedestrian trajectories from~\eqref{eq:pedestrian_predictor} incorporated as time-varying constraints to ensure collision avoidance. The cost weights are $q=1$, $r_s=0.002$, $r_\omega=0.15$, and $q_f=10$, and the input limits are $s_{\max}=0.8\,\mathrm{m/s}$ and $\omega_{\max}=0.4\,\mathrm{rad/s}$. The effective clearance $d_{\min,t}^{\mathrm{eff}}$ from~\eqref{eq:safety_margin} is applied uniformly to both the static obstacles and predicted pedestrian positions.

We compare three estimators in closed loop: EKF with nominal noise statistics, EKF with true noise statistics, and the proposed DR-EKF. Over $50$ Monte Carlo runs, the collision rates are $100\%$, $100\%$, and $66\%$, respectively. As shown in Fig.~\ref{fig:safe_nav_ped}, the DR-EKF maintains a larger and more responsive safety margin $\delta_t$ in the vicinity of pedestrians, resulting in a noticeable reduction in collision rate compared to the nominal EKF. In contrast, the nominal EKF underestimates uncertainty, leading to insufficient clearance and frequent collisions. Among the three estimators, DR-EKF achieves the lowest collision rate, indicating that explicitly accounting for linearization residuals can play a critical role in safe navigation in dynamic environments.

\section{Conclusions}\label{sec:conc}

We introduced a residual-aware DR-EKF that unifies noise uncertainty and linearization error within a Wasserstein ambiguity framework while preserving the recursive EKF structure. By reinterpreting linearization residuals as uncertainty, the method admits a tractable SDP formulation, a computable ambiguity radius, and a deterministic recursive certificate for the prior and posterior MSE of the true nonlinear system. The resulting estimator improves accuracy under model mismatch and nonlinear effects while enabling safer downstream decision-making.

Future work includes deriving tighter bounds to reduce conservatism in the effective radius and validating the framework on hardware under realistic conditions.

\appendix

\section{Proof of~\Cref{prop:sdp}}\label{app:pf_sdp}

\begin{proof}
Condition on $\mathcal Y_{t-1}$. Since $x_t=\xprior_t+e_t^-$ and $\xprior_t$ is $\mathcal Y_{t-1}$-measurable, every estimator $\psi_t(Y_t)$ can be written as
\[
\psi_t(Y_t)=\xprior_t+\phi_t(\nu_t),
\]
for some measurable map $\phi_t$, because $Y_t=(Y_{t-1},y_t)$ and $\nu_t = y_t-h(\xprior_t)-\hat{v}_t $
is a one-to-one $\mathcal Y_{t-1}$-measurable transform of $y_t$.
Under the stage-wise local linear-Gaussian surrogate, the innovation is approximated by
\[
\hat{\nu}_t = C_t e_t^- + (v_t-\hat v_t),
\]
and the corresponding stage-wise minimax problem becomes
\begin{align}
J_t :=
\inf_{\phi_t}
\sup_{\Pdist_{\epsilon,t}\in \ambset_{\epsilon,t}(\theta_t)}
\mathbb E_{\Pdist_{\epsilon,t}}
\left[
\|e_t^- - \phi_t(\hat\nu_t)\|^2
\,\middle|\,
\mathcal Y_{t-1}
\right].
\label{eq:prop-centered}
\end{align}
The stage-wise surrogate dynamics are
\begin{align}
e_t^- &= A_{t-1}e_{t-1} + (w_{t-1}-\hat w_{t-1}), \notag\\
\hat\nu_t &= C_t e_t^- + (v_t-\hat v_t).
\label{eq:prop-surrogate}
\end{align}

We first restrict to Gaussian stacked-noise laws,
\begin{align}
\begin{split}
\Pdist_{\epsilon,t}&=\Gauss(m_{\epsilon,t},\Sigma_{\epsilon,t}), \\
\Sigma_{\epsilon,t}&=
\begin{bmatrix}
\Sigma_{w,t-1} & \Sigma_{wv,t-1} \\
\Sigma_{wv,t-1}^{\top} & \Sigma_{v,t}
\end{bmatrix}
\in \psd{n_x+n_y}.
\label{eq:prop-gaussian-class}
\end{split}
\end{align}
This restriction is introduced only to derive a lower bound.
Under the surrogate, $e_{t-1}\mid\mathcal Y_{t-1}\sim\Gauss(0,\Xpostbf)$; hence, for every Gaussian choice of $\Pdist_{\epsilon,t}$, the pair $(e_{t-1},\epsilon_t)$ is jointly Gaussian, and $(e_t^-,\hat\nu_t)$ is an affine image of $(e_{t-1},\epsilon_t)$. Therefore, $(e_t^-,\hat\nu_t)$ is jointly Gaussian, and the inner infimum in~\eqref{eq:prop-centered} is attained by the conditional mean $\mathbb E[e_t^- \mid \hat\nu_t,\mathcal Y_{t-1}]$,
which is affine in $\hat\nu_t$; see, e.g.,~\cite[Propositions~3.1--3.2]{nguyen2023bridging}. Consequently, after minimizing over the affine offset, the resulting MMSE depends only on the covariance of $(e_t^-,\hat\nu_t)$.

For Gaussian laws,
\begin{align}
&W_2^2\!\left(
\Gauss(m_{\epsilon,t},\Sigma_{\epsilon,t}),
\Gauss(\hat\epsilon_t,\hat\Sigma_{\epsilon,t})
\right) =
\|m_{\epsilon,t}-\hat\epsilon_t\|^2
+
\Bures^2\!\left(\Sigma_{\epsilon,t},\hat\Sigma_{\epsilon,t}\right).
\label{eq:prop-w2-gaussian}
\end{align}
Hence, within the Gaussian subclass~\eqref{eq:prop-gaussian-class}, replacing $m_{\epsilon,t}$ by $\hat\epsilon_t$ leaves the minimized MMSE unchanged and can only decrease the Wasserstein distance. Thus, an optimizer of the Gaussian-restricted problem may be chosen with nominal mean:
\begin{align}
\Pdist_{\epsilon,t}=\Gauss(\hat\epsilon_t,\Sigma_{\epsilon,t}).
\label{eq:prop-gaussian-nominal-mean}
\end{align}

Let
\[
\underline{\lambda}_t := \lambda_{\min}(\hat\Sigma_{\epsilon,t}) > 0.
\]
Since $\hat\Sigma_{\epsilon,t}\in\pd{n_x+n_y}$, we may, without loss of optimality, restrict all subsequent covariance optimizations to
\[
\Sigma_{\epsilon,t} \succeq \underline{\lambda}_t I_{n_x+n_y},
\]
because this redundant constraint does not change the optimal value; see~\cite[Appendix A]{jang2025distributionally} and~\cite{nguyen2023bridging}.

Under~\eqref{eq:prop-gaussian-nominal-mean}, both $e_t^-$ and $\hat\nu_t$ are centered conditional on $\mathcal Y_{t-1}$. Their conditional covariances are therefore
\begin{align}
\Xprior
&=
\Cov(e_t^- \mid \mathcal Y_{t-1})
=
A_{t-1}\Xpostbf A_{t-1}^{\top} + \Sigma_{w,t-1},
\label{eq:prop-prior}
\\
T_t
&=
\Cov(e_t^-,\hat\nu_t \mid \mathcal Y_{t-1})
=
\Xprior C_t^{\top} + \Sigma_{wv,t-1},
\label{eq:prop-T}
\\
S_t
&=
\Cov(\hat\nu_t \mid \mathcal Y_{t-1}) =
C_t\Xprior C_t^{\top}
+\Sigma_{v,t}
+C_t\Sigma_{wv,t-1}
+\Sigma_{wv,t-1}^{\top}C_t^{\top}.
\label{eq:prop-S}
\end{align}

Moreover, if we define
\[
\Omega_t :=
\begin{bmatrix}
\Xprior & \Sigma_{wv,t-1}\\
\Sigma_{wv,t-1}^\top & \Sigma_{v,t}
\end{bmatrix},
\]
then
\[
\Omega_t
=
\begin{bmatrix}
A_{t-1}\Xpostbf A_{t-1}^\top & 0\\
0 & 0
\end{bmatrix}
+
\Sigma_{\epsilon,t}
\succeq
\Sigma_{\epsilon,t}
\succeq
\underline{\lambda}_t I.
\]
Therefore,
\[
S_t
=
[\,C_t\ \ I_{n_y}\,]\,
\Omega_t\,
\begin{bmatrix}
C_t^\top\\
I_{n_y}
\end{bmatrix}
\succeq
\underline{\lambda}_t\bigl(C_tC_t^\top + I_{n_y}\bigr)
\succ 0,
\]
so $S_t$ is invertible.

For a linear estimator $\phi_t(\hat\nu_t)=K_t\hat\nu_t$, the conditional MSE is
\begin{align*}
\mathcal R_t(K_t,\Sigma_{\epsilon,t})
&=
\mathbb E\left[
\|e_t^- - K_t\hat\nu_t\|^2
\,\middle|\,
\mathcal Y_{t-1}
\right] \\
&=
\Tr(\Xprior)-2\Tr(K_tT_t^{\top})+\Tr(K_tS_tK_t^{\top}) \\
&=
\Tr\!\left(\Xprior - T_tS_t^{-1}T_t^{\top}\right)
+
\Tr\!\left(
(K_t-T_tS_t^{-1})S_t(K_t-T_tS_t^{-1})^{\top}
\right).
\end{align*}
Hence, the minimum over $K_t$ is attained at
\begin{align}
K_t^* = T_tS_t^{-1},
\label{eq:prop-gain}
\end{align}
and the minimized value is
\begin{align}
\inf_{K_t}\mathcal R_t(K_t,\Sigma_{\epsilon,t})
=
\Tr\!\left(\Xprior - T_tS_t^{-1}T_t^{\top}\right).
\label{eq:prop-mmse}
\end{align}

For Gaussian laws of the form~\eqref{eq:prop-gaussian-nominal-mean}, feasibility in $\ambset_{\epsilon,t}(\theta_t)$ is equivalent to
\begin{align}
\Bures\!\left(\Sigma_{\epsilon,t},\hat\Sigma_{\epsilon,t}\right)\le \theta_t,
\end{align}
which admits the semidefinite representation
\begin{align}
\begin{bmatrix}
\hat\Sigma_{\epsilon,t} & Z_t \\
Z_t^{\top} & \Sigma_{\epsilon,t}
\end{bmatrix}
&\succeq 0, \notag\\
\Tr\bigl(\Sigma_{\epsilon,t}+\hat\Sigma_{\epsilon,t}-2Z_t\bigr)
&\le \theta_t^2.
\label{eq:prop-bures}
\end{align}
Therefore, restricting the supremum in~\eqref{eq:prop-centered} to Gaussian laws yields the lower bound
\begin{align}
J_t
\ge
\max_{\substack{\Xprior,\Sigma_{w,t-1},\Sigma_{v,t},\\
\Sigma_{wv,t-1},Z_t}}
\Tr\!\left(\Xprior - T_tS_t^{-1}T_t^{\top}\right)
\label{eq:prop-lower}
\end{align}
subject to~\eqref{eq:prop-prior}, \eqref{eq:prop-bures}, and
\begin{align}
\Sigma_{\epsilon,t}
=
\begin{bmatrix}
\Sigma_{w,t-1} & \Sigma_{wv,t-1} \\
\Sigma_{wv,t-1}^{\top} & \Sigma_{v,t}
\end{bmatrix}
\in \psd{n_x+n_y}.
\label{eq:prop-sigmaeps}
\end{align}

To upper-bound the full problem over all laws in $\ambset_{\epsilon,t}(\theta_t)$, define
\[
M_t(\theta_t):=
\left\{
(m,\Sigma):
\|m-\hat\epsilon_t\|^2 + \Bures^2(\Sigma,\hat\Sigma_{\epsilon,t}) \le \theta_t^2,\;
\Sigma \in \psd{n_x+n_y}
\right\}.
\]
By the Gelbrich inequality, every law in $\ambset_{\epsilon,t}(\theta_t)$ induces a mean--covariance pair in $M_t(\theta_t)$. Thus, restricting the estimator to be affine gives the standard Gelbrich upper bound, whereas~\eqref{eq:prop-lower} provides the corresponding Gaussian lower bound. The same conditional restricted primal--dual sandwich argument as in~\cite[Theorem~4.1 and Corollary~4.1]{nguyen2023bridging} therefore applies here: since $\Qhat_{\epsilon,t}$ is Gaussian, the image of the Gaussian Wasserstein ball under the mean--covariance projection is exactly $M_t(\theta_t)$, and the risk of any affine estimator depends on the adversary only through these moments. Hence the lower bound~\eqref{eq:prop-lower} is tight, and $J_t$ equals the maximum of
\[
\Tr\!\left(\Xprior - T_tS_t^{-1}T_t^\top\right)
\]
over the feasible covariance set. In particular, there exists a least-favorable stacked-noise law, and it may be chosen Gaussian with mean $\hat\epsilon_t$ and covariance $\Sigma_{\epsilon,t}^*$.

Finally, introduce $\Xpost$ and impose
\begin{align}
\Xpost \preceq \Xprior - T_tS_t^{-1}T_t^{\top}.
\label{eq:prop-post-bound}
\end{align}
Since $S_t\succ0$, the Schur complement gives
\begin{align}
\Xpost \preceq \Xprior - T_tS_t^{-1}T_t^{\top}
\iff
\begin{bmatrix}
\Xprior - \Xpost & T_t \\
T_t^{\top} & S_t
\end{bmatrix}
\succeq 0.
\label{eq:prop-schur}
\end{align}
Because the objective maximizes $\Tr(\Xpost)$, the inequality~\eqref{eq:prop-post-bound} is tight at optimality. Collecting the propagation constraint, the Bures LMI, the PSD block structure of $\Sigma_{\epsilon,t}$, and the Schur complement condition~\eqref{eq:prop-schur} yields the SDP~\eqref{eq:sdp_t}.

Let $\Xpostopt$ and $\Sigma_{\epsilon,t}^*$ be optimal for~\eqref{eq:sdp_t}, with associated blocks $(T_t^*,S_t^*)$. Then the optimal value of the stage-wise local linear-Gaussian surrogate of~\eqref{eq:stagewise-minimax} equals $\Tr(\Xpostopt)$, the least-favorable stacked-noise law is $\Gauss(\hat\epsilon_t,\Sigma_{\epsilon,t}^*)$, and the minimax-optimal estimator is
\begin{align}
\psi_t^*(Y_t)
=
\xprior_t + T_t^*(S_t^*)^{-1}\hat\nu_t,
\end{align}
which is exactly the innovation-form update~\eqref{eq:dr_estimator}. This proves the proposition.

\end{proof}

\section{Proof of~\Cref{lem:quad_rem_bound}}\label{app:pf_lem1}
\begin{proof}
We prove the bound for $r^f_{t-1}(e)$; the proof for $r^h_t(e)$ is identical. Fix $e$ such that the segment $\{\xpost_{t-1}+se: s\in[0,1]\}$ lies in the neighborhood of \Cref{assump:local_curv}. By the fundamental theorem of calculus,
\begin{align*}
r^f_{t-1}(e)
&=
f(\xpost_{t-1}+e,u_{t-1})
-
f(\xpost_{t-1},u_{t-1})
-
A_{t-1}e \\
= &
\int_0^1
\left(
\frac{\partial f}{\partial x}(\xpost_{t-1}+se,u_{t-1})
-
\frac{\partial f}{\partial x}(\xpost_{t-1},u_{t-1})
\right)e\,ds.
\end{align*}
Taking norms and using \Cref{assump:local_curv}, we obtain
\begin{align*}
\|r^f_{t-1}(e)\| 
&\le
\int_0^1
\left\|
\frac{\partial f}{\partial x}(\xpost_{t-1}+se,u_{t-1})
-
\frac{\partial f}{\partial x}(\xpost_{t-1},u_{t-1})
\right\|
\|e\|\,ds \\
&\le
\int_0^1 sL_f\|e\|^2\,ds =
\frac{L_f}{2}\|e\|^2.
\end{align*}
The bound for $r^h_t(e)$ follows in the same way, with $\xpost_{t-1}$ and $A_{t-1}$ replaced by $\xprior_t$ and $C_t$, and $L_f$ replaced by $L_h$.
\end{proof}

\section{Proof of~\Cref{lem:oracle_residual_radii}}\label{app:pf_lem2}
\begin{proof}
For $t\ge 1$,~\Cref{lem:quad_rem_bound} gives
\[
\|r^f_{t-1}(e_{t-1})\|^2
\le
\frac{L_f^2}{4}\|e_{t-1}\|^4,
\qquad
\|r^h_t(e_t^-)\|^2
\le
\frac{L_h^2}{4}\|e_t^-\|^4.
\]
Taking expectations and applying~\Cref{ass:fourth_moment} gives
\[
\bigl(\mathbb E[\|r^f_{t-1}(e_{t-1})\|^2]\bigr)^{1/2}
\le
\frac{L_f}{2}\alpha_f\,\mathbb E[\|e_{t-1}\|^2]
=
\eta^f_{t-1},
\]
and
\[
\bigl(\mathbb E[\|r^h_t(e_t^-)\|^2]\bigr)^{1/2}
\le
\frac{L_h}{2}\alpha_h\,\mathbb E[\|e_t^-\|^2]
=
\eta_t^h.
\]
At $t=0$, there is no process-side residual term, so $\eta^f_{-1}=0$.
\end{proof}

\section{Proof of~\Cref{lem:surrogate_posterior}}\label{app:pf_lem_surrogate}

\begin{proof}
We prove by induction that $\mathbb{E}[\|\hat e_t\|^2]\le \bar s_t^2$ for all $t\ge 0$. By the Wasserstein ambiguity assumption and $1$-Lipschitzness of coordinate projections, all marginals of $\mathcal L(\epsilon_t)$ inherit the same radius $\theta_{\epsilon,t}$.

\emph{Base case ($t=0$):}
Using the definition of $\hat e_0$, Minkowski's inequality, and \Cref{ass:det_env},
\[
\bigl(\mathbb{E}[\|\hat e_0\|^2]\bigr)^{1/2}
\le
\bar m_0\bigl(\mathbb{E}[\|e_0^-\|^2]\bigr)^{1/2}
+
\bar k_0\bigl(\mathbb{E}[\|v_0-\hat v_0\|^2]\bigr)^{1/2}.\]
By the Gelbrich inequality and the trace-tube bound~\cite[Proposition~1]{jang2025distributionally},
$\bigl(\mathbb{E}[\|e_0^-\|^2]\bigr)^{1/2}
\le
\sqrt{\Tr(\Xnom)}+\theta_{\epsilon,0}$ and
$\bigl(\mathbb{E}[\|v_0-\hat v_0\|^2]\bigr)^{1/2}
\le
\sqrt{\Tr(\hat\Sigma_{v,0})}+\theta_{\epsilon,0}$.
Substituting yields $\bigl(\mathbb{E}[\|\hat e_0\|^2]\bigr)^{1/2}\le \bar s_0$.

\emph{Induction step:}
Fix $t\ge 1$ and assume $\mathbb{E}[\|\hat e_{t-1}\|^2]\le \bar s_{t-1}^2$.
From~\eqref{eqn:ehat_pred}, Minkowski's inequality, and \Cref{ass:det_env},
\[
\bigl(\mathbb{E}[\|\hat e_t^-\|^2]\bigr)^{1/2}
\le
\bar a_{t-1}\bar s_{t-1}
+
\bigl(\mathbb{E}[\|w_{t-1}-\hat w_{t-1}\|^2]\bigr)^{1/2}.
\]
Applying the same Wasserstein bound to $w_{t-1}$ gives
$\bigl(\mathbb{E}[\|w_{t-1}-\hat w_{t-1}\|^2]\bigr)^{1/2}
\le
\sqrt{\Tr(\hat\Sigma_{w,t-1})}+\theta_{\epsilon,t}$.
Hence
\[
\bigl(\mathbb{E}[\|\hat e_t^-\|^2]\bigr)^{1/2}
\le
\bar a_{t-1}\bar s_{t-1}
+
\sqrt{\Tr(\hat\Sigma_{w,t-1})}
+
\theta_{\epsilon,t}.\]

From~\eqref{eqn:ehat}, Minkowski's inequality, and \Cref{ass:det_env},
\[
\bigl(\mathbb{E}[\|\hat e_t\|^2]\bigr)^{1/2} \! \!
\le
\bar m_t\bigl(\mathbb{E}[\|\hat e_t^-\|^2]\bigr)^{1/2} \! 
+
\bar k_t\bigl(\mathbb{E}[\|v_t-\hat v_t\|^2]\bigr)^{1/2} \! \! \leq \! \bar s_t,
\]
where we applied the same Wasserstein bound to $v_t$, thus completing the induction.
\end{proof}

\section{Proof of~\Cref{lem:onepass_upper_bound}}\label{app:pf_lem4}

\begin{proof}
We treat the cases $t=0$ and $t\ge 1$ separately. We show $\eta_{t-1}^f \le \bar\eta_{t-1}^f$ and $\eta_t^h \le \bar\eta_t^h$, which together imply the claimed bound on $\theta_{\epsilon,t}^{\mathrm{eff}}$.

\emph{Case $t=0$.}
By \Cref{ass:nominal_radius}, the $x_0$-marginal satisfies
\[
W_2(\mathcal{L}(x_0),\Qhat^-_{x,0}) \le \theta_{\epsilon,0}.\]
Applying the Wasserstein bound via the Gelbrich inequality and the trace-tube bound~\cite[Proposition~1]{jang2025distributionally} yields
\[
\mathbb{E}[\|e_0^-\|^2]
\le
\bigl(\sqrt{\Tr(\Xnom)}+\theta_{\epsilon,0}\bigr)^2
=:
\gamma_0^2.\]
Therefore,
\[
\eta_0^h
\le
\frac{L_h}{2}\alpha_h\gamma_0^2
=
\bar\eta_0^h.\]

\emph{Case $t\ge 1$.}
Since $\mathbb{E}[\|e_{t-1}\|^2]\le \bar V_{t-1}^2$, we have
$\eta^f_{t-1}
\le
\frac{L_f}{2}\alpha_f\bar V_{t-1}^2
=
\bar\eta^f_{t-1}$.
Applying Minkowski's inequality to the prior error dynamics~\eqref{eq:error-dynamics-prior} and using \Cref{ass:det_env}, we obtain
\[
\bigl(\mathbb E[\|e_t^-\|^2]\bigr)^{1/2}
\le
\bar a_{t-1}\bigl(\mathbb E[\|e_{t-1}\|^2]\bigr)^{1/2} +
\bigl(\mathbb E[\|w_{t-1}-\hat w_{t-1}\|^2]\bigr)^{1/2}
+
\bigl(\mathbb E[\|r^f_{t-1}(e_{t-1})\|^2]\bigr)^{1/2}.
\]

By \Cref{ass:nominal_radius}, the $w_{t-1}$-marginal satisfies
$W_2(\mathcal L(w_{t-1}),\Qhat_{w,t-1}) \le \theta_{\epsilon,t}$.
Applying again the Gelbrich inequality and the trace-tube bound gives
$\bigl(\mathbb E[\|w_{t-1}-\hat w_{t-1}\|^2]\bigr)^{1/2}
\le
\sqrt{\Tr(\hat\Sigma_{w,t-1})}+\theta_{\epsilon,t}$.
Combining this with $\mathbb E[\|e_{t-1}\|^2]\le \bar V_{t-1}^2$ and \Cref{lem:oracle_residual_radii}, we obtain
\[
\bigl(\mathbb E[\|e_t^-\|^2]\bigr)^{1/2}
\le
\bar a_{t-1}\bar V_{t-1}
+
\sqrt{\Tr(\hat\Sigma_{w,t-1})}
+
\theta_{\epsilon,t}
+
\bar\eta^f_{t-1}
=
\gamma_t.
\]
Hence $\mathbb E[\|e_t^-\|^2]\le \gamma_t^2$, and therefore
$\eta_t^h
\le
\frac{L_h}{2}\alpha_h\gamma_t^2
=
\bar\eta_t^h$.

Finally, substituting the bounds on $\eta_{t-1}^f$ and $\eta_t^h$ into~\eqref{eq:effective-radius-oracle} gives
$\theta_{\epsilon,t}^{\mathrm{eff}}
\le
\bar\theta_{\epsilon,t}^{\mathrm{eff}}$,
which proves the claim.
\end{proof}

\section{Proof of~\Cref{thm:dr_certificate}}\label{app:pf_thm1}
\begin{proof}
We prove part~(iii) together with the auxiliary bounds
\[
\bigl(\mathbb E[\|\delta_t^-\|^2]\bigr)^{1/2}\le \rho_t^-, \qquad\bigl(\mathbb E[\|\delta_t\|^2]\bigr)^{1/2}\le \rho_t\]
by induction on $t$, where $\delta_0^- := 0$, $\delta_0 := -K_0^*r_0^h(e_0^-)$, and, for $t\ge 1$, $\delta_t^- := e_t^- - \hat e_t^-$ and $\delta_t := e_t - \hat e_t$. Parts~(i) and~(ii) are obtained at each stage by applying \Cref{lem:onepass_upper_bound}.

For each $t\ge 0$, write $\tilde\epsilon_t=\epsilon_t+\Delta_t$, where
$\Delta_0=\begin{bmatrix}0\\ r_0^h(e_0^-)\end{bmatrix}$ and, for $t\ge 1$,
$\Delta_t=\begin{bmatrix}r_{t-1}^f(e_{t-1})\\ r_t^h(e_t^-)\end{bmatrix}$.
Using the identity coupling,
\[
W_2(\mathcal L(\tilde\epsilon_t),\mathcal L(\epsilon_t))
\le
\bigl(\mathbb E[\|\Delta_t\|^2]\bigr)^{1/2}
\le
\eta_{\epsilon,t}.\]

\emph{Base case ($t=0$).}
By \Cref{lem:onepass_upper_bound},
\[\mathbb E[\|e_0^-\|^2]\le \gamma_0^2, \qquad \eta_0^h\le \bar\eta_0^h, \qquad
\theta_{\epsilon,0}^{\mathrm{eff}}\le \bar\theta_{\epsilon,0}^{\mathrm{eff}}.\]
Hence, using \Cref{ass:nominal_radius}, \[
W_2(\mathcal L(\tilde\epsilon_0),\Qhat_{\epsilon,0})
\le
W_2(\mathcal L(\tilde\epsilon_0),\mathcal L(\epsilon_0))
+
W_2(\mathcal L(\epsilon_0),\Qhat_{\epsilon,0})
\le
\theta_{\epsilon,0}^{\mathrm{eff}}
\le
\bar\theta_{\epsilon,0}^{\mathrm{eff}}.\]
Therefore $\mathcal L(\tilde\epsilon_0)\in \ambset_{\epsilon,0}(\bar\theta_{\epsilon,0}^{\mathrm{eff}})$, proving~(i), while $\mathbb E[\|e_0^-\|^2]\le \gamma_0^2$ proves~(ii).

Moreover,
\[
\bigl(\mathbb E[\|\delta_0\|^2]\bigr)^{1/2}
\le
\bar k_0\bigl(\mathbb E[\|r_0^h(e_0^-)\|^2]\bigr)^{1/2}
\le
\bar k_0\,\bar\eta_0^h
=
\rho_0\] and $\bigl(\mathbb E[\|\delta_0^-\|^2]\bigr)^{1/2}=0=\rho_0^-$.
By \Cref{lem:surrogate_posterior}, $\mathbb E[\|\hat e_0\|^2]\le \bar s_0^2$.
Since $e_0=\hat e_0+\delta_0$, Minkowski's inequality gives
$\bigl(\mathbb E[\|e_0\|^2]\bigr)^{1/2}
\le
\bar s_0+\rho_0
=
\bar V_0$.
Hence 
\[
\mathbb E[\|e_0\|^2]\le \bar V_0^2,
\]
proving~(iii).

\emph{Induction step.}
Fix $t\ge 1$ and assume
$\mathbb E[\|e_{t-1}\|^2]\le \bar V_{t-1}^2$
and
$\bigl(\mathbb E[\|\delta_{t-1}\|^2]\bigr)^{1/2}\le \rho_{t-1}$.
Applying \Cref{lem:onepass_upper_bound} yields
\[
\eta^f_{t-1}\le \bar\eta^f_{t-1},\qquad
\mathbb E[\|e_t^-\|^2]\le \gamma_t^2, \qquad \eta_t^h\le \bar\eta_t^h,\]
as well as
$\theta_{\epsilon,t}^{\mathrm{eff}}\le \bar\theta_{\epsilon,t}^{\mathrm{eff}}$.
The bound $\mathbb E[\|e_t^-\|^2]\le \gamma_t^2$ proves~(ii).

Moreover, using \Cref{ass:nominal_radius} and the triangle inequality,
$
W_2(\mathcal L(\tilde\epsilon_t),\Qhat_{\epsilon,t})
\le
\bar\theta_{\epsilon,t}^{\mathrm{eff}}$.
Thus $\mathcal L(\tilde\epsilon_t)\in \ambset_{\epsilon,t}(\bar\theta_{\epsilon,t}^{\mathrm{eff}})$, proving~(i).

From the error decompositions,
$\delta_t^- = A_{t-1}\delta_{t-1}+r_{t-1}^f(e_{t-1})$ and
$\delta_t = (I-K_t^*C_t)\delta_t^- - K_t^*r_t^h(e_t^-)$.
Applying Minkowski's inequality, \Cref{ass:det_env}, the induction hypothesis, and the bounds above, we obtain
\[
\bigl(\mathbb E[\|\delta_t^-\|^2]\bigr)^{1/2}
\le
\bar a_{t-1}\rho_{t-1}+\bar\eta_{t-1}^f
=
\rho_t^- \]
and 
\[
\bigl(\mathbb E[\|\delta_t\|^2]\bigr)^{1/2}
\le
\bar m_t\,\rho_t^-+\bar k_t\,\bar\eta_t^h
=
\rho_t.\]
Finally, \Cref{lem:surrogate_posterior} gives $\mathbb E[\|\hat e_t\|^2]\le \bar s_t^2$. Since $e_t=\hat e_t+\delta_t$, Minkowski's inequality yields
\[
\bigl(\mathbb E[\|e_t\|^2]\bigr)^{1/2}
\le
\bar s_t+\rho_t
=
\bar V_t.\]
Hence 
\[
\mathbb E[\|e_t\|^2]\le \bar V_t^2,\] proving~(iii).
\end{proof}

% \section{Implementation details of DR-EKF}

%%%%%%%%%%% %%%%%%%%%%%%%%%%%%%%%%%%%%%%%%%%%%%%%%%%%%%%%%%%%%%%%%%%%%%%%%%%%%%%%
\bibliographystyle{IEEEtran}
\bibliography{ref}

 \vspace{-0.15in}
\end{document}